\documentclass[a4paper,11pt]{article}
\pdfoutput=1 

\usepackage{jheppub} 

\usepackage[T1]{fontenc} 

\usepackage{indentfirst} 
\usepackage{amsfonts,amssymb,amsmath}
\usepackage{bm}          
\usepackage{url}
\usepackage{enumitem}
\usepackage[usenames,dvipsnames]{xcolor}
\usepackage{todonotes}
\usepackage{multicol}
\usepackage{MnSymbol}
\usepackage{tensor}
\usepackage{hyperref}
\usepackage{todonotes}

\usetikzlibrary{arrows}

\AtBeginDocument{}%

\newcommand{\corurl}{red}
\newcommand{\corcite}{ForestGreen}
\newcommand{\corlink}{blue}

\newcommand{\dd}{\mathrm{d\!l}}

\newcommand{\E}{\mathsf{E}}
\newcommand{\F}{\mathsf{F}}
\newcommand{\D}{\mathsf{D}}

\newtheorem{prop}{Proposition}

\hypersetup{linktocpage,colorlinks,urlcolor=\corurl,citecolor=\corcite,linkcolor=\corlink,
pdftitle={Hamiltonian GNH analysis of the parametrized unimodular extension of the Holst action},
pdfauthor={J.F. Barbero, B. Daz, J. Margalef-Bentabol, E.J.S. Villase\~nor}}

\def\QED{{\boldmath$\rule{0.5em}{0.5em}$}}                                
\def\markatright#1{\leavevmode\unskip\nobreak\quad\hspace*{\fill}{#1}}    
\def\qed{\markatright{\QED}}
\newenvironment{proof}[1][Proof]{\noindent\textbf{#1.} }{\qed\\}

\title{\boldmath Hamiltonian Gotay-Nester-Hinds analysis of the parametrized unimodular extension of the Holst action }

\author[a,b]{J. Fernando Barbero G.,}
\author[c,d,b]{Bogar D\'{\i}az,}
\author[e,b]{Juan Margalef-Bentabol}
\author[c,b]{and Eduardo J. S. Villase\~nor\,}

\affiliation[a]{Instituto de Estructura de la Materia, CSIC. Serrano 123, 28006 Madrid, Spain}
\affiliation[b]{Grupo de Teor\'{\i}as de Campos y F\'{\i}sica Estad\'{\i}stica. Instituto Gregorio Mill\'an (UC3M). Unidad Asociada al Instituto de Estructura de la Materia, CSIC}
\affiliation[c]{Departamento de Matem\'aticas, Universidad Carlos III de Madrid. Avda.\  de la Universidad 30, 28911 Legan\'es, Spain.}
\affiliation[d]{Departamento de F\'{\i}sica de Altas Energ\'{\i}as, Instituto de Ciencias Nucleares, Universidad Nacional Aut\'onoma de M\'exico, Apartado Postal 70-543, Ciudad de M\'exico, 04510, M\'exico}
\affiliation[e]{Institute for Gravitation and the Cosmos \& Physics Department, Penn State University, University Park, PA 16802, USA}

\emailAdd{fbarbero@iem.cfmac.csic.es}
\emailAdd{bodiazj@math.uc3m.es}
\emailAdd{juanmargalef@psu.edu}
\emailAdd{ejsanche@math.uc3m.es}

\abstract{
We give a detailed account of the Hamiltonian GNH analysis of the pa\-ra\-me\-tri\-zed unimodular extension of the Holst action. The purpose of the paper is to derive, through the clear geometric picture furnished by the GNH method, a simple Hamiltonian formulation for this model and explain why it is difficult to arrive at it in other approaches. We will also show how to take advantage of the field equations to anticipate the simple form of the constraints that we find in the paper.}

\arxivnumber{}

\begin{document}
\maketitle
\flushbottom

\section{Introduction}{\label{sec_introduction}}

Although less popular than the Dirac algorithm \cite{Dirac}, the Gotay-Nester-Hinds (GNH) approach to the Hamiltonian formulation of mechanical systems and field theories defined by singular Lagrangians is very powerful and conceptually clean \cite{GNH1,GNH2,GNH3,BPV,margalef2018thesis}. Its geometric underpinnings provide a rigorous viewpoint that avoids many of the drawbacks of Dirac's method --in particular when applied to field theories-- while ultimately giving the same basic information. Several differences between both approaches should be noted:
\begin{itemize}
\item[i)] Dirac's method relies heavily on the language of classical mechanics. For instance, singular Lagrangian systems are characterized as those for which it is impossible to write all the velocities in terms of momenta; this leads to the ensuing appearance of constraints and the need to enforce their stability in order to guarantee the consistency of time evolution. The GNH method, on the other hand, is based on geometry. For instance, the notion of dynamical stability is translated into the requirement that the vector fields that encode the dynamics must be tangent to the phase space submanifold defined by the constraints. Although the rationale behind Dirac's approach can be rephrased in geometric terms (see \cite{GNH2,Diracnos,nos_Pontryagin}), this is non-standard and probably feels unnatural for many readers.
\item[ii)] The final descriptions provided by both methods are different, although it is possible in practice to go back and forth from one to the other. Dirac's method is designed to produce a Hamiltonian description in \emph{the full phase space}. This is useful for quantization because the canonical symplectic structure is retained and, hence, the idea of turning Poisson brackets into commutators can be implemented as in standard quantum mechanics. The presence of first class constraints is taken into account by using their quantized version to select \emph{physical subspaces} of the full Hilbert space of the system, while second class constraints are taken care of either by solving them or using the so-called Dirac brackets. The arena of the GNH approach is the primary constraint submanifold endowed with a presymplectic form obtained by pulling back the canonical symplectic structure of the full phase space. Even though the setting is slightly different, the constraints obtained with the Dirac method can also be found and, conceivably, quantized in a similar way.  
\item[iii)] From a practical point of view, the emphasis on geometry characteristic of the GNH method has some unexpected consequences. In particular, the possibility of altogether avoiding the use of Poisson brackets when dealing with the tangency conditions mentioned in i) is instrumental in circumventing the difficulties that crop up, for instance, when field theories are defined in spatial regions with boundary. Another consequence of the shift in perspective is the possibility of incorporating the, often subtle, functional analytic issues relevant for field theories that originate in the fact that their configuration spaces are infinite dimensional manifolds.
\item[iv)] Finally, it must be pointed out that the differences between both methods sometimes lead to insights within one of them that are difficult to arrive at in the other. In fact, this paper illustrates an instance of this phenomenon.
\end{itemize}

The main purpose of this work is to apply the GNH method to the study of the Holst \cite{Holst} action and some interesting generalizations of it, in particular, its parametrized unimodular version. Unimodular gravity is an alternative approach to general relativity with some interesting features, in particular regarding the role of the cosmological constant. Its Hamiltonian analysis in metric variables is well known \cite{Teitel}. However, and despite some claims to the effect \cite{Yamashita}, a similar analysis in terms of tetrads starting from the Holst action has not been performed yet. The Holst action has several features that make the study of its parametrized unimodular version quite attractive. In particular, it involves the Immirzi parameter and leads to the real Ashtekar formulation. On its turn, parametrized unimodular gravity is interesting because parametrization offers some useful insights on the problem of time \cite{Kuchar-1,UW}. Given the differences between the metric and tetrad formulations for general relativity, we deem it interesting to understand the Hamiltonian formulation of parametrized unimodular general relativity in the context of the Holst action. The hope --ultimately realized-- is that the analysis will provide an interesting perspective on the Hamiltonian formulation of unimodular gravity. As we will see, the final description given by the GNH approach --one of the results of the present paper-- is concise and clean (see \cite{nos} for a short partial summary involving just the Holst action). It naturally leads to the real Ashtekar formulation of general relativity and illuminates some issues related to the role of the time gauge and the Immirzi parameter, both at the classical and quantum level \cite{nos}.

The paper is structured as follows. After this introduction, we devote section \ref{sec_Action} to a discussion of the action principle used in the paper (the parametrized, unimodular version of the Holst action), the field equations, and the Lagrangian formulation. Section \ref{sec_Hamiltonian} contains a detailed discussion of the GNH analysis. The very simple final form of the constraints in the Hamiltonian formulation that we find suggests a streamlined approach to the Hamiltonian treatment of field theories linear in first order time derivatives. We discuss it in section \ref{sec_streamlined}. The literature on the Hamiltonian treatment of the Holst action and how the Ashtekar formulation can be derived from it is quite extensive. In order to put in perspective the results presented here we provide an appraisal of the main works on this subject in section \ref{sec_appraisal}. Although it is difficult to be exhaustive, we do try to provide a balanced assessment of the most important results and their relation to the present work. In section \ref{sec_conclusions} we give our conclusions and some comments. The contrast between the simple Hamiltonian formulation that we find here and the long computations necessary to arrive at it is somehow striking. We discuss this issue in the conclusions. The paper ends with several appendices where we give a number of auxiliary results and some computational details.

%
%
%
\section{Variational setting}\label{sec_Action}

\subsection{Action and equations of motion}\label{subsec_action}

We consider the following generalization of the Holst action
\[
S({\bm{e}},{\bm{\omega}},\Lambda,\Phi)=\int_{\mathbb{R}\times\Sigma}\Big(\big(\ast({\bm{e}}^I\wedge {\bm{e}}^J)+\frac{\varepsilon}{\gamma}{\bm{e}}^I\wedge {\bm{e}}^J\big)\wedge {\bm{F}}_{IJ}+\Lambda\big(\Phi^*\mathsf{vol}-\frac{1}{12}\epsilon_{IJKL}{\bm{e}}^I\wedge {\bm{e}}^J\wedge {\bm{e}}^K\wedge {\bm{e}}^L\big)\Big)\,.
\]
In this expression $\Sigma$ is a closed (i.e. compact without boundary), orientable, 3-dimensional manifold. This implies that $\Sigma$ is parallelizable and, hence, globally-defined frames exist. The cotetrads ${\bm{e}}^I$ are 1-forms, ${\bm{\omega}}^I_{\,\,\,J}$ is a $\mathfrak{so}(1,3)$-valued connection 1-form with curvature ${\bm{F}}^{I}_{\,\,\,J}:=\mathrm{d}{\bm{\omega}}^I_{\,\,\,J}+{\bm{\omega}}^I_{\,\,\,K}\wedge{\bm{\omega}}^K_{\,\,\,\,\,J}$. It satisfies the identity $\bm{DF}^I_{\,\,\,J}=0$, where $\bm{D}$ is defined by suitably extending ${\bm{D}}\alpha^I=\mathrm{d}\alpha^I+{\bm{\omega}}^I_{\,\,\,J}\wedge \alpha^J$. 
The tetrads are required to be non-degenerate, i.e.,  $\epsilon_{IJKL}{\bm{e}}^I\wedge {\bm{e}}^J\wedge {\bm{e}}^K\wedge {\bm{e}}^L$ is a volume form in $\mathbb{R}\times\Sigma$. They also have to satisfy the condition that, for all $\tau\in\mathbb{R}$, the hypersurfaces $\{\tau\}\times\Sigma$ are spacelike as measured by the metric $\bm{e}_I\otimes \bm{e}^I$. The Levi-Civita symbol $\epsilon_{IJKL}$ is totally antisymmetric and chosen to satisfy $\epsilon_{0123}=+1$. In more abstract terms, it should be interpreted as a volume form in $\mathfrak{so}(1,3)$.  The internal indices $I,J, \dots$ take the values $0, 1, 2, 3$. When needed, these indices will be raised and lowered with the invariant metric $\eta$ in $\mathfrak{so}(1,3)$ that, in an appropriate basis, can be written as $\eta=\mathrm{diag}(\varepsilon, +1,+1,+1)$, with $\varepsilon= -1$. We have included $\varepsilon$ to keep track of the spacetime signature and facilitate the extension of our results to the Euclidean case. The dual $\ast$ of $V^{IJ}$ is
\[
\ast V^{IJ}:=\frac{1}{2}\epsilon^{IJ}_{\,\,\,\,\,\,\,\,KL}V^{KL}\,.
\] 

Let $\mathcal{M}$ be a 4-dimensional manifold diffeomorphic to $\mathbb{R}\times\Sigma$. The volume form $\mathsf{vol}$ on $\mathcal{M}$ is defined by a fixed, background (i.e., non-dynamical) metric $g$ on $\mathcal{M}$. The inclusion of this fiducial metric is not necessary but it will allow us to reuse some computations from \cite{parame_scalar,EM}, in which case we have to restrict slightly the diffeomorphisms $\Phi\in\mathrm{Diff}(\mathbb{R}\times\Sigma,\mathcal{M})$. Indeed, we consider the diffeomorphisms such that for all $\tau\in\mathbb{R}$, $\{\tau\}\times\Sigma$ is a $\Phi^*\!g$-spacelike hypersurface. The action depends on $\Phi$ only through the $\Phi^*\mathsf{vol}$ term.

The scalar field $\Lambda\in C^\infty(\mathbb{R}\times\Sigma)$ is \emph{dynamical}  and not the cosmological constant (at least at this stage). It plays the role of a Lagrange multiplier enforcing the \emph{parametrized unimodularity condition}
    \begin{equation}\label{uni-cond}
    \frac{1}{12}\epsilon_{IJKL}{\bm{e}}^I\wedge {\bm{e}}^J\wedge {\bm{e}}^K\wedge {\bm{e}}^L=\Phi^*\mathsf{vol}\,.
    \end{equation}

Notice that the action is defined on $\mathbb{R}\times\Sigma$. This notwithstanding, as $\mathcal{M}$ is diffeomorphic to $\mathbb{R}\times\Sigma$, it is straightforward to change the viewpoint and define it on $\mathcal{M}$, in which case the dynamical diffeomorphisms would take $\mathcal{M}$ to another manifold $\mathcal{N}$.

The Immirzi parameter is denoted as $\gamma$ ($\neq0$). It is convenient to introduce the invariant $SO(1,3)$  tensor
\begin{equation}\label{PIJKL}
P_{IJKL}:=\frac{1}{2}\left(\epsilon_{IJKL}+\frac{\varepsilon}{\gamma}\eta_{IK}\eta_{JL}-\frac{\varepsilon}{\gamma}\eta_{JK}\eta_{IL}\right)\,.
\end{equation}
Its main properties, including the form of its inverse $[P^{-1}]^{IJKL}$ (which exists only if $\gamma^2\neq \varepsilon$) can be found in Appendix \ref{appendix_useful_results}.

By using $P_{IJKL}$ we can rewrite the action as
\begin{equation}\label{action}
S({\bm{e}},{\bm{\omega}},\Lambda,\Phi)=\int_{\mathbb{R}\times\Sigma}\Big(P_{IJKL}{\bm{e}}^I\wedge {\bm{e}}^J\wedge {\bm{F}}^{KL}+\Lambda\big(\Phi^*\mathsf{vol}-\frac{1}{12}\epsilon_{IJKL}{\bm{e}}^I\wedge {\bm{e}}^J\wedge {\bm{e}}^K\wedge {\bm{e}}^L\big)\Big)\,.
\end{equation}

The field equations are obtained by varying this action with respect to the dynamical variables ${\bm{e}}^I$, ${\bm{\omega}}^I_{\,\,\,J}$, $\Lambda$ and $\Phi$. The variations with respect to the tetrads ${\bm{e}}^I$ and the connection ${\bm{\omega}}^I_{\,\,\,J}$ give the equations
\begin{subequations}
\begin{align}
&{\bm{e}}^{[I}\!\wedge {\bm{De}}^{J]}=0\,,\label{e2}\\
&2P_{IJKL}{\bm{e}}^J\wedge {\bm{F}}^{KL}-\frac{1}{3}\Lambda \epsilon_{IJKL}{\bm{e}}^J\wedge {\bm{e}}^K\wedge {\bm{e}}^L=0\,.\label{e1}
\end{align}
\end{subequations}

The variations with respect to $\Lambda$ give the unimodularity condition \eqref{uni-cond}. As the field equations of parametrized theories imply that the dynamical diffeomorphisms are always arbitrary, this condition only restricts the possible values of the tetrads.

Finally, the variations with respect to the dynamical diffeomorphisms $\Phi$ (which we know from the parametrization procedure that do not furnish any 
additional conditions) can be written in terms of Lie derivatives $\pounds_{V_\Phi}$ [the variation of a diffeomorphism $\Phi:\mathbb{R}\times\Sigma\rightarrow \mathcal{M}$ can be represented by a vector field $V_\Phi\in\mathfrak{X}(\mathcal{M})$] as
\begin{equation}\label{Zvariation}
D_{(\Phi,V_\Phi)}S=-\int_{\mathcal{M}}\pounds_{V_\Phi}(\Phi^{-1*}\Lambda)\mathsf{vol}\,.
\end{equation}
The vanishing of the integral in \eqref{Zvariation} for every $V_\Phi$  implies that $\mathrm{d}(\Phi^{-1*}\Lambda)=0$ and, hence, $\Phi^{-1*}\mathrm{d}\Lambda=0$. Since $\Phi$ is a diffeomorphism we conclude that $\mathrm{d}\Lambda=0$. This last field equation tells us that $\Lambda$ is actually a constant (an \emph{integration constant} as usually stated in the traditional literature on this subject, see \cite{Kuchar-1,Teitel}). As a consequence, and given its role in \eqref{e1}, $\Lambda$ becomes a cosmological constant, of an arbitrary magnitude, through a dynamical mechanism. The equation $\mathrm{d}\Lambda=0$ is, as we mentioned before, redundant because it can be obtained from \eqref{e1} by taking
\begin{equation}\label{eq: consecuencia Lambda=cte}
0={\bm{D}}\big(P_{IJKL}{\bm{e}}^J\wedge {\bm{F}}^{KL}-\frac{1}{3}\Lambda \epsilon_{IJKL}{\bm{e}}^J\wedge {\bm{e}}^K\wedge {\bm{e}}^L\big)\,,
\end{equation}
using the identity ${\bm{DF}}_{IJ}=0$, the fact that, for non-degenerate tetrads, \eqref{e2} is equivalent to ${\bm{De}}^I=0$ (see \cite{Cattaneo1}), and the non-degeneracy of $\epsilon_{IJKL}{\bm{e}}^I\wedge {\bm{e}}^J\wedge {\bm{e}}^k\wedge {\bm{e}}^L$.

Notice that computing the covariant differential of ${\bm{De}}^I=0$ we get the identity ${\bm{F}}^I_{\,\,\,J}\wedge {\bm{e}}^J=0$, hence, plugging this into \eqref{e1}, the $\gamma$-dependent terms drop out and \eqref{e2} and \eqref{e1} turn into the field equations given by the usual Hilbert-Palatini action with a cosmological constant.

\subsection{Lagrangian formulation}\label{subsec_lagrangian}

In order to obtain the Lagrangian from the action \eqref{action}, we need to perform a 3+1 decomposition. The manifold $\mathbb{R}\times\Sigma$ is naturally foliated by the 3-dimensional hypersurfaces $\{\tau\}\times \Sigma$ with $\tau\in \mathbb{R}$. The tangent vectors to the parametrized curves $c_p:\mathbb{R}\rightarrow \mathbb{R}\times \Sigma:\tau\mapsto (\tau,p)$, with $p\in \Sigma$, define a vector field $\partial_\tau\in\mathfrak{X}(\mathbb{R}\times\Sigma)$. For each $\tau\in\mathbb{R}$ we introduce the embedding $\jmath_\tau:\Sigma\rightarrow \mathbb{R}\times\Sigma:p\mapsto (\tau,p)$.

In terms of these geometric elements, we can write for any 4-form $\mathcal{L}$
\[
\int_{\mathbb{R}\times\Sigma}\mathcal{L}=\int_\mathbb{R}\mathrm{d}\tau\int_\Sigma \jmath_\tau^*\imath_{\partial \tau}\mathcal{L}\,.
\]
If the preceding integral is the action for a particular field theory, $\mathcal{L}$ is a \emph{Lagrangian 4-form} and the real function
\[
\tau\mapsto\int_\Sigma\jmath_\tau^*\imath_{\partial \tau}\mathcal{L}
\]
can often be interpreted as being determined by a Lagrangian $L:TQ\rightarrow \mathbb{R}$ (a real function in the tangent bundle $TQ$ of a configuration space $Q$) evaluated on curves in $Q$.

Given differential forms of arbitrary degree in $\mathbb{R}\times\Sigma$, it is convenient to build other differential forms adapted to the foliation defined by the Cartesian product $\mathbb{R}\times\Sigma$. If ${\bm{\alpha}}\in\Omega^p(\mathbb{R}\times\Sigma)$, with $p=0,\ldots,4$, we define its transverse and tangent parts
\begin{align}
\begin{split}
&\alpha_{\mathrm{t}}:=\imath_{\partial_\tau}{\bm{\alpha}}\in\Omega^{p-1}(\mathbb{R}\times\Sigma)\,,\\
&\alpha:={\bm{\alpha}}-\mathrm{d}\tau\wedge\alpha_{\mathrm{t}}\in\Omega^p(\mathbb{R}\times\Sigma)\,,\\
\end{split}\label{decomposition}
\end{align}
leading to the decomposition ${\bm{\alpha}}=\alpha+\mathrm{d}\tau\wedge\alpha_{\mathrm{t}}$. Notice that one can also perform the former decomposition using the normal to the foliation $n$, as in \cite{parame_scalar,EM}, in which case the Hamiltonian turns out to be zero everywhere in phase space. However, in the present example, it is easier to break the objects with the fixed foliation using $\partial_\tau$. Besides, in this case, the comparison with the unparametrized version is straightforward \cite{nos}. Obviously, $\displaystyle \imath_{\partial_\tau}\alpha_{\mathrm{t}}=0$ and $\displaystyle \imath_{\partial_\tau}\alpha=0$. The basic objects used to find the Lagrangian corresponding to the action \eqref{action} are obtained with the help of the previous decomposition for ${\bm{e}}^I$ and ${\bm{\omega}}^I_{\phantom{I}J}$. We also need to find out how to perform a 3+1 decomposition of the dynamical diffeomorphisms $\Phi$ and the type of dynamical objects obtained by doing this. This is explained in appendix \ref{appendix_diffeos}.

The first result coming from the 3+1 decomposition is the characterization of the configuration space of the system. In the present case, it consists of the scalar fields $e_{\mathrm{t}}^{\,I}\,,\omega_{\mathrm{t}\phantom{I}J}^{\,I}\,,\Lambda\in C^\infty(\Sigma)$, the 1-forms $e^I\,,\omega^I_{\phantom{I}J}\in\Omega^1(\Sigma)$, and the $g$-spacelike embeddings $X:\Sigma\hookrightarrow\mathcal{M}$. The points in the tangent bundle of the configuration space of our system are denoted as ${\bm{\mathrm{v}}_q}$ (where $q=(e^I_{\mathrm{t}},e^I,\omega^{IJ}_{\mathrm{t}}, \omega^{IJ},\Lambda,X)$ denotes a point in $Q$) with components $(v_{e_\mathrm{t}}, v_e, v_{\omega_\mathrm{t}},v_\omega,v_\Lambda, v_X)$ that can be interpreted as velocities. We have $v_{e_\mathrm{t}}\,,v_{\omega_\mathrm{t}}\,,v_\Lambda\in C^\infty(\Sigma)$, $v_e\,,v_\omega\in\Omega^1(\Sigma)$, and $v_X\in\Gamma(X^*T\mathcal{M})$ (i.e. the velocity associated with the embedding $X$ is a vector field along the map $X$). In terms of these objects the Lagrangian can be written as
\begin{align}
L ({\bm{\mathrm{v}}_q})=\int_\Sigma&\Big[P^{IJ}_{\quad KL}\big(2\,e_{\mathrm{t}}^K e^L\!\wedge F_{IJ}+e^K\!\wedge e^L\!\wedge(v_{\omega IJ}-D\omega_{\mathrm{t}IJ})\big)\label{Lagrangian}\\
&\hspace*{2mm}+\varepsilon \Lambda n_X(v_X)\mathsf{vol}_{\gamma_X}-\frac{1}{3}\Lambda\epsilon_{IJKL}e_{\mathrm{t}}^I(e^J\wedge e^K \wedge e^L)\Big]\,.\nonumber
\end{align}
where the curvature is $F_I^{\,\,J}:=\mathrm{d}\omega_I^{\,\,J}+\omega_I^{\,\,K}\wedge\omega_{K}^{\,\,J}$, and the embedding-dependent objects $n_X$ and $\gamma_X$ are defined in appendix \ref{appendix_diffeos}. As we can see, the Lagrangian only depends on the $v_\omega$ and $v_X$ components of the velocity. Notice also that $L$ is linear in $v_\omega$ and $v_X$, a fact that it is not completely obvious \emph{a priori} as far as the dynamical diffeomorphisms are concerned.

%
%
%

\section{Hamiltonian approach}\label{sec_Hamiltonian}

\subsection{Momenta and Hamiltonian} The canonical momenta are obtained from the fiber derivative $FL:TQ\rightarrow T^*Q$ determined by the Lagrangian $L$. They are\allowdisplaybreaks
\begin{subequations}\label{momenta}
\begin{align}
{\bm{\mathrm{p}}_{e_{\mathrm{t}}}}(w_{e_{\mathrm{t}}}^I)&:=\langle FL({\bm{\mathrm{v}}_q}),(w_{e_{\mathrm{t}}}^I,0,0,0,0,0)\rangle&&=0\,,\\
{\bm{\mathrm{p}}_e}(w_e^I)&:=\langle FL({\bm{\mathrm{v}}_q}),(0,w_{e}^I,0,0,0,0)\rangle&&=0\,,\\
{\bm{\mathrm{p}}_{\omega_\mathrm{t}}}(w_{\omega_\mathrm{t}}^{IJ})&:=\langle FL({\bm{\mathrm{v}}_q}),(0,0,w_{\omega_\mathrm{t}}^{IJ},0,0,0)\rangle&&=0\,,\\
{\bm{\mathrm{p}}_\omega}(w_\omega^{IJ})&:=\langle FL({\bm{\mathrm{v}}_q}),(0,0,0,w_{\omega}^{IJ},0,0)\rangle&&=\int_\Sigma P_{IJKL}w_\omega^{IJ}\wedge e^K\wedge e^L\,,\\
{\bm{\mathrm{p}}_\Lambda}(w_\Lambda)&:=\langle FL({\bm{\mathrm{v}}_q}),(0,0,0,0,w_\Lambda,0)\rangle&&=0\,,\\
{\bm{\mathrm{p}}_X}(w_X)&:=\langle FL({\bm{\mathrm{v}}_q}),(0,0,0,0,0,w_X)\rangle&&=\int_\Sigma \varepsilon n_X(w_X)\Lambda \mathsf{vol}_{\gamma_X}\,,
\end{align}
\end{subequations}
where we denote $\langle\bm{\mathrm{p}},{\bm{\mathrm{w}}}\rangle:=\bm{\mathrm{p}}({\bm{\mathrm{w}}})$.

In the present case the fiber derivative $FL$ is obviously not onto. As it is not a diffeomorphism between $TQ$ and $T^*Q$, our action defines a singular system.

As we can see, the momenta are all independent of the velocities. The conditions \eqref{momenta} can be interpreted as \emph{primary constraints} that characterize the image of $FL$, usually known as the primary constraint submanifold $\mathfrak{M}_0$. As all the momenta can be written in terms of the configuration variables $e_{\mathrm{t}}^I$, $e^I$, $\omega_{\mathrm{t}}^{IJ}$, $\omega^{IJ}$, $\Lambda$, and $X$, the submanifold $\mathfrak{M}_0$ can be parametrized by these objects, in fact, $\mathfrak{M}_0$ \emph{is the configuration space} of the system. For this reason we will often refer to the configuration variables when talking about points in $\mathfrak{M}_0$.

On $\mathfrak{M}_0$ the Hamiltonian is defined by the condition $H\circ FL=E$, where the energy is $E:TQ\rightarrow \mathbb{R}:{\bm{\mathrm{v}}_q}\mapsto\langle FL({\bm{\mathrm{v}}_q}),{\bm{\mathrm{v}}_q}\rangle-L({\bm{\mathrm{v}}_q})$. In the present case we have
\begin{equation}\label{Hamiltonian}
H=\int_\Sigma\left(P_{IJKL}\Big(e^I\!\wedge e^J\!\wedge D\omega_{\mathrm{t}}^{KL}-2\,e_{\mathrm{t}}^I e^J\!\wedge F^{KL}\Big)+\frac{1}{3}\Lambda\epsilon_{IJKL}\, e_{\mathrm{t}}^I(e^J\!\wedge e^K\!\wedge e^L)\right)\,.
\end{equation}
Notice that $H$ does not depend either on the momenta (it is defined only on $\mathfrak{M}_0$) or the embeddings $X$.

\subsection{GNH analysis}

In order to get a Hamiltonian description for the dynamics of a singular Lagrangian system, it is necessary to identify a maximal submanifold $\mathfrak{M}$ of the primary constraint submanifold $\mathfrak{M}_0$ and vector fields $\mathbb{Z}\in\mathfrak{X}(\mathfrak{M}_0)$ tangent to $\mathfrak{M}$ such that the following condition holds
\begin{equation}\label{eq_ham_vfields}
(\imath_{\mathbb{Z}}\omega-\dd H)|\mathfrak{M}=0\,.
\end{equation}
In the previous expression $\omega$ denotes the pullback of the canonical symplectic form $\Omega$ in $T^*Q$ to $\mathfrak{M}_0$ and $\dd$ denotes the differential in phase space. A very good way to solve \eqref{eq_ham_vfields} is to follow the procedure introduced by Gotay, Nester, and Hinds in \cite{GNH1,GNH2,GNH3}.

Vector fields in the phase space discussed here have the form
\begin{equation}\label{vector_field}
\mathbb{Z}=\big(Z_{e_{\mathrm{t}}},Z_{e},Z_{\omega_{\mathrm{t}}},Z_\omega,Z_\Lambda,Z_X; {\bm{\mathrm{Z}}}_{e_{\mathrm{t}}},{\bm{\mathrm{Z}}}_e,{\bm{\mathrm{Z}}}_{\omega_{\mathrm{t}}},{\bm{\mathrm{Z}}}_\omega,{\bm{\mathrm{Z}}}_\Lambda,{\bm{\mathrm{Z}}}_X \big) \,,
\end{equation}
where the boldface components are associated with the ``momenta directions'' in phase space and the components in the ``field directions'' have the following internal index structure
\[
\big(Z_{e_{\mathrm{t}}}^I,Z_{e}^I,Z_{\omega_{\mathrm{t}}}^{IJ},Z_\omega^{IJ},Z_\Lambda,Z_X\big)\,.
\]
Given $\mathbb{Z},\mathbb{Y}\in\mathfrak{X}(T^*Q)$ and the canonical symplectic form $\Omega$ we have
\begin{align}
\Omega(\mathbb{Z},\mathbb{Y})=&{\bm{\mathrm{Y}}}_{e_{\mathrm{t}}}(Z_{e_{\mathrm{t}}})-{\bm{\mathrm{Z}}}_{e_{\mathrm{t}}}(Y_{e_{\mathrm{t}}})+
                               {\bm{\mathrm{Y}}}_e (Z_e)-{\bm{\mathrm{Z}}}_e (Y_e)+
                               {\bm{\mathrm{Y}}}_{\omega_{\mathrm{t}}} (Z_{\omega_{\mathrm{t}}})-{\bm{\mathrm{Z}}}_{\omega_{\mathrm{t}}} (Y_{\omega_{\mathrm{t}}})\nonumber\\
                              +&{\bm{\mathrm{Y}}}_\omega (Z_\omega)-{\bm{\mathrm{Z}}}_\omega (Y_\omega)+
                               {\bm{\mathrm{Y}}}_\Lambda (Z_\Lambda)-{\bm{\mathrm{Z}}}_\Lambda (Y_\Lambda)+
                               {\bm{\mathrm{Y}}}_X (Z_X)-{\bm{\mathrm{Z}}}_X (Y_X)\label{symplectic_form} \,.
\end{align}
Remember that $\imath_{\mathbb{Z}}\Omega(\mathbb{Y})=\Omega(\mathbb{Z},\mathbb{Y})$. From here on we will work on $\mathfrak{M}_0$. The pullback of $\Omega$ to the primary constraint submanifold $\mathfrak{M}_0$ is given by (see appendix \ref{Ap_GNH_pullback})
\begin{align}
&\omega(\mathbb{Z},\mathbb{Y})=\label{omega_pullback}\\
&\int_\Sigma 2P_{IJKL}\big(Z_\omega^{IJ}\wedge Y_e^K-Y_\omega^{IJ}\wedge Z_e^K\big)\wedge e^L+\big(Z_{\mathsf{X}}^\perp(Y_\Lambda-\imath_{Y_{\mathsf{X}}^\top}\mathrm{d}\Lambda)
-Y_{\mathsf{X}}^\perp(Z_\Lambda-\imath_{Z_{\mathsf{X}}^\top}\mathrm{d}\Lambda)\big)\mathsf{vol}_{\gamma_X}\nonumber \,.
\end{align}
On the other hand
\begin{align}
\dd H(\mathbb{Y})=&\int_\Sigma \Big[Y_\omega^{IJ}\wedge\big(-2D(P_{IJKL}e_{\mathrm{t}}^{\,K} e^L)+2P_{IKLM}\omega_{\mathrm{t}J}^{\phantom{\mathrm{t}J}K}e^L\wedge e^M\big)\label{dH}\\
&\hspace*{7mm}+Y_e^I \wedge \big( 2P_{IJKL}\big(e^J\wedge D\omega_{\mathrm{t}}^{\phantom{\mathrm{t}}KL}+e_{\mathrm{t}}^{\,J} F^{KL}\big)-\Lambda \epsilon_{IJKL}e_{\mathrm{t}}^{\,J} e^K\wedge e^L \big)\nonumber\\
&\hspace*{7mm}-Y_{\omega_{\mathrm{t}}}^{IJ}D\big(P_{IJKL}e^K\wedge e^L\big)+\frac{1}{3}Y_\Lambda \epsilon_{IJKL}e_{\mathrm{t}}^{\,I} (e^J\wedge e^K\wedge e^L)\nonumber\\
&\hspace*{7mm}+Y_{e_{\mathrm{t}}}^I\big(-2P_{IJKL}e^J\wedge F^{KL}+\frac{1}{3}\Lambda \epsilon_{IJKL}(e^J\wedge e^K\wedge e^L)\big)\Big]\nonumber \,.
\end{align}
Requiring now that $\omega(\mathbb{Z},\mathbb{Y})=\dd H(\mathbb{Y})$ for all vector fields $\mathbb{Y}\in\mathfrak{X}(\mathfrak{M}_0)$ we get the following:

\medskip

\noindent 1) Conditions involving the components of $\mathbb{Z}$.
\begin{subequations}
\begin{align}
&Z_{\mathsf X}^\perp\,\mathsf{vol}_{\gamma_X}=\frac{1}{3}\epsilon_{IJKL}e_{\mathrm{t}}^I(e^J\wedge e^K \wedge e^L)\,,\label{condition1}\\
&Z_\Lambda=\imath_{Z_{\mathsf X}^\top}\mathrm{d}\Lambda\,,\label{condition2}\\
&Z_{\mathsf X}^\perp \mathrm{d}\Lambda=0\,,\label{condition3}\\
&Z_e^{[I}\wedge e_{\phantom{e}}^{J]}=D(e_{\mathrm{t}}^{[I}e_{\phantom{e}}^{J]})-[P^{-1}]^{IJ}_{\,\,\,\,\,\,\,MN}P^{MK}_{\quad\,\,\,\,LP}\,\omega^{\,N}_{{\mathrm{t}}\,\,\,K}\,e^L\wedge e^P\,,\label{condition4}\\
&2P_{IJKL}e^J\wedge Z_{\omega}^{\,KL}=2P_{IJKL}\big(e^J\wedge D\omega_{\mathrm{t}}^{KL}+e_{\mathrm{t}}^{\,J}F^{KL}\big)-\Lambda\,\epsilon_{IJKL}e_{\mathrm{t}}^{\,J} e^K\wedge e^L \,. \label{condition5}
\end{align}
\end{subequations}
\noindent 2) Secondary constraints
\begin{subequations}
\begin{align}
& e^{[I}\!\wedge D e^{J]}=0\,,\label{constraint1}\\
& P_{IJKL}e^J\wedge F^{KL}-\frac{1}{3!}\Lambda\,\epsilon_{IJKL}e^{J}\wedge e^K\wedge e^L=0\,.\label{constraint2}
\end{align}
\end{subequations}
Before we analyze in detail the conditions on $\mathbb{Z}$ and the secondary constraints we make a couple of comments. First, in order to find the components of the Hamiltonian vector field $\mathbb{Z}$ we have to solve equations \eqref{condition1}-\eqref{condition5}. Their solutions will give us the components of $\mathbb{Z}$ in terms of the configuration variables. As the equations are inhomogeneous, they may not be solvable in the whole of $\mathfrak{M}_0$. If this is the case, new secondary constraints will arise that we will duly have to take into account. It may also happen that some of the components of $\mathbb{Z}$ are left arbitrary (a feature characteristic of gauge theories). Indeed, it is straightforward to see that this is the case since there are no conditions involving $Z^I_{e_{\mathrm{t}}}$ or $Z_{\omega_{\mathrm{t}}}^{IJ}$. Second, as we are only interested in the values of $\mathbb{Z}$ on the final constraint submanifold $\mathfrak{M}$, we can take advantage of the constraints to simplify the expressions for the components of $\mathbb{Z}$. This is specially useful when checking the tangency of $\mathbb{Z}$ to $\mathfrak{M}$.

\subsubsection{Conditions on the components of $\mathbb{Z}$}\label{subsub_Z}

The conditions \eqref{condition1}-\eqref{condition3} are easy to analyze. To begin with, it is important to point out that diffeomorphisms must be interpreted as curves of embeddings in this setting (see appendix \ref{appendix_diffeos}), hence, we must demand that $Z_{\mathsf X}^\perp\neq 0$ at every point in $\Sigma$. This can be easily achieved by restricting the configuration variables to satisfy
\begin{equation}\label{condition_ev}
\epsilon_{IJKL}e_{\mathrm{t}}^{\,I}(e^J\wedge e^K \wedge e^L)\neq0\,,
\end{equation}
everywhere on $\Sigma$ because then, \eqref{condition1} implies $Z_{\mathsf X}^\perp\neq 0$ for all $p\in\Sigma$. Notice that, roughly speaking, \eqref{condition_ev} defines an open subset of our initial configuration space. A further consequence of \eqref{condition_ev} is that \eqref{condition3} is then equivalent to $\mathrm{d}\Lambda=0$ --a secondary constraint-- which means that the scalar field $\Lambda$ has the same value on all the points of $\Sigma$ although, in principle, it could depend on the evolution parameter. Notice, however, that condition \eqref{condition2} implies $Z_\Lambda=0$ and, hence,  $\dot{\Lambda}=Z_\Lambda=0$, i.e. $\Lambda$ must be also \emph{constant under evolution}. We then conclude that $\Lambda$ is constant in $\mathcal{M}$ in agreement with the result obtained from the field equations. In fact, this result is also a consequence of the dynamics of the theory and can be found without invoking the parametrization (in analogy with \eqref{eq: consecuencia Lambda=cte}). As we can see, the parametrized unimodularity condition can be easily implemented in the connection-triad formalism discussed here and leads to conclusions similar to those of \cite{Kuchar-1}. In summary, if we restrict our configuration variables to satisfy \eqref{condition_ev}, the conditions \eqref{condition1}-\eqref{condition3} are equivalent to
\begin{subequations}
\begin{align}
&Z_X^\perp=\frac{1}{3}\left(\frac{\epsilon_{IJKL}e_{\mathrm{t}}^{\,I}(e^J\wedge e^K \wedge e^L)}{\mathsf{vol}_{\gamma_X}}\right)\,,\label{condition1-bis}\\
&\mathrm{d}\Lambda=0\,,\label{condition2-bis}\\
&Z_\Lambda=0\,,\label{condition3-bis}
\end{align}
\end{subequations}
where we have used the notation introduced in appendix \ref{appendix_notation}.

Let us discuss now conditions \eqref{condition4} and \eqref{condition5}. As we can see, both of them can be interpreted as linear inhomogeneous equations for the 1-forms $Z_e^I$ and $Z_\omega^{IJ}$, respectively. The fact that they are inhomogeneous means that they may not be solvable in all of $\mathfrak{M}_0$. In other words, additional secondary constraints may appear.

By using identity \eqref{identity1} of appendix \ref{appendix_useful_results}, condition \eqref{condition4} can be written in the following form, independent of $\gamma$ and $\varepsilon$:
\begin{equation}\label{condition4-bis}
{Z^{[I}}_{\!\!\!\!e}\wedge e^{J]}=D({e^{\,\,[I}}_{\!\!\!\!\!\!\!\!\!\mathrm{t}}\,\,\,\,e^{J]})-{\omega_{\mathrm{t}K}}^{[I}e^{J]}\!\wedge e^K\,.
\end{equation}
These are six equations for the four 1-forms $Z_e^I$, ($I=0,\ldots,3$) or, counting components, 18 equations for 12 unknowns. It is convenient to write them in the form
\begin{equation}\label{equation_xi}
\Xi^{[I}\wedge e^{J]}={e^{\,\,[I}}_{\!\!\!\!\!\!\!\!\!\mathrm{t}}\,\,\,\,De^{J]}\,,
\end{equation}
with
\[
\Xi^{I}:=Z_e^I-De_{\mathrm{t}}^I-e^K\omega_{\mathrm{t}K}^{\,\,\,\,\,\,\,\,I}\,.
\]
In order to solve these equations it helps to split them in two groups: one corresponding to $I=i$ and $J=j$ with $i\,,j=1,\ldots,3$ and the other to $I=0$ and $J=j$ with $j=1,\ldots,3$. We thus get
\begin{subequations}
\begin{align}
&\epsilon_{ijk}e^j\wedge \Xi^k=\epsilon_{ijk}A^{jk}\,,\label{ec_Xi}\\
&\Xi^0\wedge e^j-\Xi^j\wedge e^0={A^j}\,,\label{ec_Xi0}
\end{align}
\end{subequations}
with
\begin{subequations}
\begin{align}
&A^{ij}:={e_{\mathrm{t}}}^{[i}De^{j]}\,\label{As_1} \,,\\
&A^j:={e_{\mathrm{t}}}^0 De^j-{e_{\mathrm{t}}}^j De^0\,.\label{As_2}
\end{align}
\end{subequations}
It is important to notice that in the previous expressions $De^i$ means
\[
De^i=\mathrm{d}e^i+\omega^i_{\,\,J}\wedge e^J=\mathrm{d}e^i+\omega^i_{\,\,j}\wedge e^j+\omega^i_{\,\,0}\wedge e^0 \,,
\]
(analogously, for other objects of this type).

Equation \eqref{ec_Xi} can always be solved for $\Xi^k$ without having to impose any conditions on the inhomogeneous term (see appendix \ref{appendix_useful_results}). However, equation \eqref{ec_Xi0} can only be solved for $\Xi^0$ when the following condition holds [here $\Xi^i$ is the solution to \eqref{ec_Xi}]
\begin{equation}\label{Newconstraint_1}
\Xi^i\wedge e^0\wedge e^j+A^i\wedge e^j+\Xi^j\wedge e^0\wedge e^i+A^j\wedge e^i=0\,,
\end{equation}
(see appendix \ref{appendix_useful_results}). This gives the new secondary constraints
\begin{equation}\label{secondary_1}
\epsilon_{klm}e_{\mathrm{t}}^k\big(\D^{il}\E^{jm}+\D^{jl}\E^{im}\big)-e_{\mathrm{t}}^0\big(\D^{ij}+\D^{ji}\big)+e_{\mathrm{t}}^i \D^{\circ j}+e_{\mathrm{t}}^j \D^{\circ i}=0 \,,
\end{equation}
in terms of the objects introduced in appendix \ref{appendix_notation}.

Let us discuss now how to solve for $Z_\omega^{IJ}$ in \eqref{condition5} and, in particular, whether new constraints arise in the process. In terms of components we have now 12 equations and 18 unknowns. By defining
\[
T_{IJ}:=P_{IJKL}\big(-Z_\omega^{KL}+D\omega_{\mathrm{t}}^{KL}\big)
\]
the equations \eqref{condition5} become
\begin{equation}\label{condition5-bis}
e^J\wedge T_{IJ}=-P_{IJKL}e_{\mathrm{t}}^J F^{KL}+\frac{1}{2}\Lambda \epsilon_{IJKL} e_{\mathrm{t}}^J e^K\wedge e^L\,.
\end{equation}
In the following we will denote the r.h.s. of the previous equation as $C_I$. In components, equations \eqref{condition5-bis} are
\begin{subequations}
\begin{align}
&e^i\wedge T_{0i}=C_0\,,\label{ec_T0i}\\
&e^j\wedge T_{ij}=C_i+e^0\wedge T_{0i}\,,\label{ec_Tij}
\end{align}
\end{subequations}
with
\begin{align}
&C_0=-\frac{1}{2}\epsilon_{ijk}e_{\mathrm{t}}^i F^{jk}-\frac{\varepsilon}{\gamma}e_{\mathrm{t}}^i F_{0i}+\frac{1}{2}\Lambda\epsilon_{ijk}e_{\mathrm{t}}^i e^j\wedge e^k\,,\label{Cs}\\
&C_i=\frac{1}{2}\epsilon_{ijk}e_{\mathrm{t}}^0 F^{jk}+\varepsilon\cdot\epsilon_{ijk}e_{\mathrm{t}}^jF^{k}_{\,\,\,\,0}+\frac{\varepsilon}{\gamma}e_{\mathrm{t}}^0 F_{0i}+\frac{\varepsilon}{\gamma}e_{\mathrm{t}}^jF_{ji}+\Lambda\epsilon_{ijk}e_{\mathrm{t}}^j\,e^0\wedge e^k-\frac{1}{2}\Lambda\epsilon_{ijk}e_{\mathrm{t}}^0\,e^j\wedge e^k\,.\nonumber
\end{align}
As shown in appendix \ref{appendix_useful_results}, equations \eqref{ec_T0i} and \eqref{ec_Tij} can always be solved with no conditions coming from their inhomogeneous terms, so no new secondary constraints appear here. Notice that to solve the latter we have to dualize and consider $T_{ij}=\epsilon_{ijk}T^k$.

\subsubsection{Simplifying the constraints}\label{subsub_simplif_constraints}
\noindent Before solving the equations for the components of $\mathbb{Z}$ and checking the tangency of $\mathbb{Z}$ to the submanifold of $\mathfrak{M}_0$ determined by the secondary constraints, it is useful to look at these constraints in detail and simplify them as much as possible.

The constraints \eqref{constraint1} can be split in the following set of conditions
\begin{subequations}
\begin{align}
&\epsilon_{ijk}e^j\wedge De^k=0\,,\label{constraint1_ij}\\
&De^0\wedge e^i-De^i\wedge e^0=0\,.\label{constraint1_i0}
\end{align}
\end{subequations}
In terms of the objects introduced in appendix \ref{appendix_notation}, these are equivalent to
\begin{subequations}
\begin{align}
&\D^{ij}-\D^{ji}=0\,,\label{constraint1_ij_bis}\\
&\D^{\circ i}-\D^{i\circ}=0\,.\label{constraint1_i0_bis}
\end{align}
\end{subequations}

It is very important to point out that, when \eqref{constraint1_ij_bis} and \eqref{constraint1_i0_bis} hold, the secondary constraints \eqref{secondary_1} are equivalent to (see appendix \ref{Ap_GNH_rewritingnew})
\begin{equation}\label{secondary_2}
\left(e_{\mathrm{t}}^0-\frac{1}{2}\epsilon_{klm}e_{\mathrm{t}}^k\E^{lm}\right)\D^{(ij)}=0 \,.
\end{equation}
As discussed in appendix \ref{Ap_non_deg}, the non-degeneracy condition for the tetrads \eqref{condition_ev} implies that the term in parentheses in \eqref{secondary_2} is different from zero at every point of $\Sigma$, hence, the new secondary constraints can be written in the pleasingly concise form
\begin{equation}\label{secondary_4}
\D^{(ij)}=0\,.
\end{equation}
We prove now an important result:
\begin{prop}
The constraints $\D^{(ij)}=0$ and $e^{[I}\wedge De^{J]}=0$ are equivalent to the condition $De^I=0$.
\end{prop}

\begin{proof}

\noindent $\displaystyle \boxed \Leftarrow$ This is obvious as $De^i=0$ implies $\displaystyle\D^{ij}:=\left(\frac{e^i\wedge De^j}{w}\right)=0$.

\medskip

\noindent $\displaystyle \boxed \Rightarrow$ To begin with, notice that, as a consequence of \eqref{constraint1_ij_bis} and \eqref{constraint1_i0_bis}, the conditions $\D^{(ij)}=0$ together with $e^{[I}\wedge De^{J]}=0$ are equivalent to $\D^{ij}=0$ and $\D^{i\circ}-\D^{\circ i}=0$. Now, according to \eqref{rel_7}, $\D^{ij}=0$ implies $De^i=0$. From this and the definition of $\D^{i\circ}$, we immediately obtain $\D^{i\circ}=0$ so that \eqref{constraint1_i0_bis} implies $\D^{\circ i}=0$ and, hence, $De^0=0$ as a consequence of \eqref{rel_8}.
\end{proof}

\medskip

\noindent As a trivial --but nonetheless reassuring-- check, notice that \eqref{constraint1} and \eqref{secondary_4} provide 12 conditions ``per point'', the same as $De^I=0$. Notice also that we have shown that $\D^{\circ i}=\D^{i\circ}=0$ on the secondary constraint submanifold.

The condition $De^I=0$ can be used to simplify the constraints \eqref{constraint2}. To this end, notice that $De^I=0$ implies $F^I_{\phantom{I}J}\wedge e^J=0$ and, hence, the $\gamma$-dependent terms in \eqref{constraint2} vanish. Summarizing, we have shown that the secondary constraints found up to this point --ultimately the whole set of constraints-- can be written as
\begin{subequations}\label{constraints_simple}
\begin{align}
&\mathrm{d}\Lambda=0\,,\label{constraints_simple1}\\
&De^I=0\,,\label{constraints_simple2}\\
&\epsilon_{IJKL}e^J\wedge\left(F^{KL}-\frac{1}{3}\Lambda e^K\wedge e^L\right)=0\,.\label{constraints_simple3}
\end{align}
\end{subequations}

In terms of the objects introduced in appendix \ref{appendix_notation} they can be written as
\begin{subequations}\label{constraints_sf}
\begin{align}
&\mathrm{d}\Lambda=0\,,\\
&\D^{ij}=0\,,\\
&\D^{\circ i}=\D^{i\circ}=0\,,\\
&2\Lambda-\epsilon_{ijk}\F^{ijk}=0\,,\\
&\epsilon_{ijk}(\Lambda\E^{jk}+2\F^{j\circ k}-\F^{\circ jk})=0\,.
\end{align}
\end{subequations}

\subsubsection{Tangency conditions}\label{subsub_tangency}

\noindent Up to this point, we have restricted ourselves to study the secondary constraints coming either from \eqref{eq_ham_vfields} or as conditions for the solvability of the equations for the components of $\mathbb{Z}$. There is, though, an additional consistency requirement which is central to the GNH approach: we must ensure the tangency of $\mathbb{Z}$ to the submanifold of $\mathfrak{M}_0$ defined by the secondary constraints. The tangency conditions can be easily obtained by computing the directional derivatives of the constraints in the direction of $\mathbb{Z}$. As a side remark, it is important to notice that this step \emph{does not involve} the presymplectic form $\omega$ (in other words, Poisson brackets play no role here). The tangency conditions take the form of additional linear and homogeneous equations involving the components of $\mathbb{Z}$ and the variables $e_{\mathrm{t}}^I$, $\omega^I_{\mathrm{t}\,\,J}$, $e^I$, $\omega^I_{\phantom{I}J}$, $\Lambda$ and $X$. In the present case it is useful to derive them from \eqref{constraints_simple}. They are
\begin{subequations}\label{tan}
\begin{align}
&\mathrm{d}Z_\Lambda=0 \,, \label{tan_1}\\
&DZ_e^I+Z_{\omega J}^I\wedge e^J=0 \,, \label{tan_2}\\
&\epsilon_{IJKL}\left(Z_e^J\wedge (F^{KL}-\Lambda e^K\wedge e^L)+e^J\wedge DZ_\omega^{KL}-\frac{1}{3}Z_\Lambda e^J\wedge e^K\wedge e^L\right)=0\,.\label{tan_3}
\end{align}
\end{subequations}

These must be considered together with \eqref{condition5}, \eqref{condition3-bis}, \eqref{condition4-bis}, and taking into account that the constraints \eqref{constraints_simple} must hold.

\subsubsection{Final consistency analysis}\label{subsub_Final_consistency}

So far, we have found a set of constraints \eqref{constraints_simple}, defining a submanifold $\mathfrak{M}$ of $\mathfrak{M}_0$ where the dynamical variables are forced to live, and the conditions \eqref{condition5}, \eqref{condition3-bis}, \eqref{condition4-bis}, and \eqref{tan} that the components of the restriction of the Hamiltonian vector fields to $\mathfrak{M}$ must satisfy. The latter have different origins: some of them come directly from the resolution of  $\omega(\mathbb{Z},\mathbb{Y})=\dd H(\mathbb{Y})$ on $\mathfrak{M}_0$ whereas the rest appear as tangency conditions. Notice that, despite their different origins, we have to consider all these equations together in order to get the final form of the components of the Hamiltonian vector field that defines the dynamics of our model.

There are several possibilities now:
\begin{itemize}
\item There are no solutions for the components of $\mathbb{Z}$ so the theory is inconsistent (obviously not the case here).
\item The equations can be solved with no extra conditions on the configuration variables. Their solutions then give the components of the Hamiltonian vector field that encodes the dynamics on the final constraint submanifold.
\item New consistency conditions appear. These should be added to the secondary constraints together with the corresponding tangency conditions to start the process again.
\end{itemize}

Let us find out what happens in the present case. To begin with, we immediately see that \eqref{condition3-bis} implies \eqref{tan_1} and we can remove $Z_\Lambda$ from the remaining equations. Next, it is convenient to solve \eqref{condition4-bis} for $Z_e^I$ or, equivalently, \eqref{equation_xi}. When the constraints hold these equations can be written as
\begin{subequations}
\begin{align}
&\epsilon_{ijk}e^j\wedge \Xi^k=0\,,\label{ec_Xi_const}\\
&\Xi^0\wedge e^j-\Xi^j\wedge e^0=0\,.\label{ec_Xi0_const}
\end{align}
\end{subequations}

Now, by using the results of appendix \ref{appendix_useful_results} we immediately see that the (unique) solution of \eqref{ec_Xi_const} and \eqref{ec_Xi0_const} is $\Xi^I=0$, which tells us that
\begin{equation}\label{Ze_const}
Z_e^I=De_{\mathrm{t}}^I-\omega^I_{\mathrm{t}\,J}e^J\,.
\end{equation}
We can use now \eqref{Ze_const} to simplify \eqref{condition5}, \eqref{tan_2}, and \eqref{tan_3}. To this end we $D$-differentiate \eqref{Ze_const} and use the constraints to get
\[
DZ_e^I=F^I_{\phantom{I}J}e_{\mathrm{t}}^J-D\omega^I_{\mathrm{t}\,J}\wedge e^J\,,
\]
which, plugged into the tangency condition \eqref{tan_2} gives
\begin{equation}\label{tan_2_bis}
F^I_{\phantom{I}J}e_{\mathrm{t}}^J+(Z^I_{\omega J}-D\omega^I_{\mathrm{t}\,J})\wedge e^J=0\,.
\end{equation}
When this condition holds \eqref{condition5} is equivalent to
\begin{equation}\label{condition5_bis}
\epsilon_{IJKL}\Big(e^J\wedge\big(Z_\omega^{KL}-D\omega_{\mathrm{t}}^{KL}\big)-e_{\mathrm{t}}^J\big(F^{KL}-\Lambda e^K\wedge e^L\big)\Big)=0\,.
\end{equation}
Remember that we also have the tangency condition
\begin{equation}\label{tan3_bis}
\epsilon_{IJKL}\left(Z_e^J\wedge (F^{KL}-\Lambda e^K\wedge e^L)+e^J\wedge DZ_\omega^{KL}\right)=0\,.
\end{equation}
At this point the only task left is to solve \eqref{tan_2_bis}, \eqref{condition5_bis} and \eqref{tan3_bis} for $Z_\omega^{IJ}$. In order to do this we first show that by $D$-differentiating \eqref{condition5_bis}, using \eqref{Ze_const} and the secondary constraints, equation \eqref{tan3_bis} holds. This is a direct computation that we give in some detail in appendix \ref{Ap_GNH_Zomega}. Now, in order to find the $Z_\omega^{IJ}$ we only have to consider \eqref{tan_2_bis} and \eqref{condition5_bis}. 
An important point that we have to address in the first place is the consistency of these equations: can they always be solved or should we introduce extra secondary constraints in order to guarantee their solvability?

As we have shown above, the constraints $De^I=0$ are equivalent to $e^{(i}\wedge De^{j)}=0$ and $e^{[I}\wedge De^{J]}=0$, hence, the tangency conditions of the Hamiltonian vector field to the secondary constraint submanifold can also be written in the form
\begin{subequations}\label{TanDe_comp}
\begin{align}
&e^{(i}\wedge D Z^{j)}_e+e^{(i}\wedge Z^{j)}_{\omega\,\,K}\wedge e^K=0\,,\label{TanDe_comp1}\\
&e^{[I}\wedge D Z^{J]}_e+e^{[I}\wedge Z^{J]}_{\omega\,\,K}\wedge e^K=0\,,\label{TanDe_comp2}
\end{align}
\end{subequations}
which, on account of \eqref{Ze_const}, become
\begin{subequations}\label{TanDe_compS}
\begin{align}
&{e^{(i}\wedge F^{j)}}_{\!K} e_{\mathrm{t}}^K+e^{(i}\wedge S^{j)}_{\phantom{j}\,K}\wedge e^K=0\,,\label{TanDe_compS1}\\
&e^{[I}\wedge F^{J]}_{\phantom{j}\,\,\,K}e_{\mathrm{t}}^K+e^{[I}\wedge S^{J]}_{\phantom{j}\,\,\,K}\wedge e^K=0\,,\label{TanDe_compS2}
\end{align}
\end{subequations}
where we have introduced the notation $S^I_{\phantom{I}J}:=Z^I_{\omega J}-D\omega^I_{\mathrm{t}J}$. By defining $F_\Lambda^{IJ}:=F^{IJ}-\Lambda e^I\wedge e^J$, condition \eqref{condition5_bis} can be written as
\[
e^{[J}\wedge S^{KL]}-e_{\mathrm{t}}^{[J}F_\Lambda^{KL]}=0\,,
\]
which implies
\begin{equation}\label{55mod}
e_J\wedge \big(e^{[J}\wedge S^{KL]}-e_{\mathrm{t}}^{[J}F_\Lambda^{KL]}\big)=0\,.
\end{equation}
Taking into account that $e_J\wedge F_\Lambda^{JK}:=e_J\wedge(F^{JK}-\Lambda e^J\wedge e^K)=0$ --remember that the constraint $De^I=0$ implies $F^I_{\phantom{J}J}\wedge e^J=0$-- \eqref{55mod} is simply
\[
2e^{[K}\wedge S^{L]}_{\phantom{L]}J}\wedge e^J-e_{\mathrm{t}J}e^J\wedge F^{KL}+\Lambda e_{\mathrm{t}J}e^J\wedge e^K\wedge e^L=0\,,
\]
which, as a consequence of the constraint \eqref{constraints_simple3}, can be immediately seen to be equivalent to \eqref{TanDe_compS2}. We then conclude that the problem of finding the components of $Z^I_{\omega J}$ on the secondary constraint submanifold reduces to that of solving \eqref{condition5_bis} together with \eqref{TanDe_compS1}.

\medskip

\noindent Let us look now at equation \eqref{condition5_bis}. By separately considering $I=0$ and $I=i$, it can be split into
\begin{subequations}\label{Eqs55SM}
\begin{align}
&e^i\wedge S_i=\frac{1}{2}\epsilon_{ijk}e_{\mathrm{t}}^iF_\Lambda^{jk}\,,\label{Eqs55SM1}\\
&\epsilon_{ijk}e^j\wedge S^{0k}=e^0\wedge S_i+\epsilon_{ijk}e_{\mathrm{t}}^jF_\Lambda^{0k}-\frac{1}{2}\epsilon_{ijk}e_{\mathrm{t}}^0F_\Lambda^{jk}\,,\label{Eqs55SM2}
\end{align}
\end{subequations}
where $2S^i:=\epsilon^{ijk}S_{jk}$. As shown in appendix \ref{appendix_useful_results}, equation \eqref{Eqs55SM1} can always be solved and the solution written in the form
\begin{equation}\label{Solution_Si}
S_i=\tau_{ij}e^j+\sigma_i \,,
\end{equation}
where $\tau_{ij}\in C^\infty(\Sigma)$ with $\tau_{ij}=\tau_{ji}$ but, otherwise arbitrary, and $\sigma_i$ a concrete function of the dynamical fields which can be computed by using equation \eqref{sols3}. Plugging $S_i$ into \eqref{Eqs55SM2} and using \eqref{sols1} we find
\begin{equation}\label{Solution_Soi}
S^{0k}=\tau_{ij}\E^{kj}e^i+\eta^k\,,
\end{equation}
where the concrete form of $\eta^k$ is not specially illuminating so we do not give it here [of course, it can be obtained by using \eqref{sols1}]. As we can see, equation \eqref{condition5_bis} can always be solved but its solutions depend on the arbitrary objects $\tau_{ij}$. At this point, all that is left to do is plugging \eqref{Solution_Si} and \eqref{Solution_Soi} into \eqref{TanDe_compS1} and study the resulting equations for $\tau_{ij}$. A straightforward computation tells us that these equations are
\begin{equation}\label{equationtau}
\big(\delta^i_k\delta^j_l-\delta^{ij}\delta_{kl}+\varepsilon \E_k^{\,\,(i}\E^{j)}_{\phantom{j\,\,}l} \big)\tau^{kl}+\xi^{kl}=0\,,
\end{equation}
where $\xi^{kl}$ is another concrete function of the dynamical fields. Now, as we show in appendix \ref{Ap_det}, the $6\times6$ matrix
\begin{equation}\label{Matrix}
{M^{(ij)}}_{\!(kl)}:=\delta^i_k\delta^j_l-\delta^{ij}\delta_{kl}+\varepsilon \E_k^{\,\,(i}\E^{j)}_{\,\,\,\,\,l} \,,
\end{equation}
is always invertible for the field configurations that we are considering in the paper and, hence, it is always possible to solve for all the components of $\tau_{ij}$. One these are known, we can plug them into the expressions for $S_i$ and $S^{0k}$ and, finally, obtain $Z_{\omega}^{IJ}$. The complete expressions are long and not specially illuminating, so we will not give them here. Of course they are simpler in the time gauge.

%
%
%

\section{A streamlined approach to the GNH analysis of the Holst action}\label{sec_streamlined}

\noindent Here we show how the constraints and Hamiltonian vector fields obtained above by using the GNH method can be directly obtained from the field equations (see \cite{GN,Batlle}). This is interesting because knowing in advance that extra secondary constraints are expected and having an idea of their form can be useful. As discussed in section \ref{sec_Action}, the field equations of the $4$-dimensional action \eqref{action} are equivalent to
\begin{subequations}
\begin{align}
&{\bm{De}}^I=0\,,\label{fieldeq_De}\\
&\epsilon_{IJKL}{\bm{e}}^J\wedge\left({\bm{F}}^{KL}-\frac{1}{3}\Lambda {\bm{e}}^K\wedge{\bm{e}}^L\right)=0\,,\label{fieldeq_F}\\
&\Phi^*\mathsf{vol}-\frac{1}{12}\epsilon_{IJKL}{\bm{e}}^I\wedge {\bm{e}}^J\wedge {\bm{e}}^K\wedge {\bm{e}}^L=0\,,\label{fieldeq_unimodularity}\\
&\mathrm{d}\Lambda=0\,,\label{fieldeq_Lambda}
\end{align}
\end{subequations}
where the last equation is redundant as explained in section \ref{sec_introduction}.
Let us write
\begin{align*}
&{\bm{e}}^I=\underline{e}^I+\mathrm{d}\tau\,e_{\mathrm{t}}^I\,,\\
&{\bm{\omega}}^{I}{}_J=\underline{\omega}^I_{\phantom{I}J}+\mathrm{d}\tau\,\omega^{\,\,I}_{\mathrm{t}\phantom{I}J}\,.
\end{align*}
with $e_{\mathrm{t}}^I:=\imath_{\partial_\tau}{\bm{e}}^I$ and $\omega_{\mathrm{t}\,J}^I:=\imath_{\partial_\tau}{\bm{\omega}}^{I}{}_J$. Plugging this decomposition into \eqref{fieldeq_De} gives
\[
0={\bm{De}}^I=\mathrm{d}\underline{e}^I+\underline{\omega}^I_{\phantom{I}J}\wedge\underline{e}^J+\big(\mathrm{d}e^I_{\mathrm{t}}+\underline{\omega}^I_{\phantom{I}J} e^J_{\mathrm{t}}-\omega^I_{\mathrm{t}\,J}\underline{e}^J\big)\wedge \mathrm{d}\tau\,.
\]
Pulling this back to $\Sigma_\tau$, we get (remember that $\jmath_\tau^*\mathrm{d}\tau=0$)
\begin{equation}\label{lig1}
\mathrm{d}\jmath_\tau^*\underline{e}^I+\jmath_\tau^*\underline{\omega}^I_{\phantom{I}J}\wedge\jmath_\tau^*\underline{e}^J=0\,,
\end{equation}
and pulling back to $\Sigma_\tau$ after taking the interior product with $\imath_{\partial_\tau}$ we find
\begin{equation}\label{campo1}
\jmath_\tau^*\pounds_{\partial_\tau}\underline{e}^I=\frac{\mathrm{d}}{\mathrm{d}\tau}(\jmath_\tau^*\underline{e}^I)=\mathrm{d}\jmath_\tau^*e^I_{\mathrm{t}}+\jmath_\tau^*\underline{\omega}^I_{\phantom{I}J}\jmath_\tau^* e^J_{\mathrm{t}}-\jmath_\tau^*\omega^I_{\mathrm{t}\,J}\jmath_\tau^*\underline{e}^J\,.
\end{equation}
By performing the substitutions
\begin{align}
&\jmath_\tau^*\underline{e}^I\rightarrow e^I\,,&&\jmath_\tau^*e_{\mathrm{t}}^I\rightarrow e_{\mathrm{t}}^I\,,&&\frac{\mathrm{d}}{\mathrm{d}\tau}(\jmath_\tau^*\underline{e}^I)\rightarrow Z^I_e\,,\label{sust1}\\
&\jmath_\tau^*\underline{\omega}^I_{\phantom{I}J}\rightarrow \omega^I_{\phantom{I}J}\,,&&\jmath_\tau^* \omega^I_{\mathrm{t}\,J}\rightarrow \omega^I_{\mathrm{t}\,J} \,,\nonumber
\end{align}
in \eqref{lig1} and \eqref{campo1} we get
\begin{align*}
&De^I=0\,,\\
&Z_e^I=De^I_{\mathrm{t}}-\omega^I_{\mathrm{t}\,J}e^J\,.
\end{align*}
An analogous computation for \eqref{fieldeq_F} gives
\begin{align*}
&\hspace*{-1.5cm}\epsilon_{IJKL}{\bm{e}}^J\wedge\left({\bm{F}}^{KL}-\frac{1}{3}\Lambda {\bm{e}}^K\wedge {\bm{e}}^L\right)\\
&=\epsilon_{IJKL}\,\underline{e}^J\wedge\left(\mathrm{d}\underline{\omega}^{KL}+\underline{\omega}^K_{\phantom{K}M}\wedge \underline{\omega}^{ML}-\frac{1}{3}\Lambda\underline{e}^K\wedge\underline{e}^L\right)\\
&+\epsilon_{IJKL}\,e^J_{\mathrm{t}}\mathrm{d}\tau\wedge\left(\mathrm{d}\underline{\omega}^{KL}+\underline{\omega}^K_{\phantom{K}M}\wedge \underline{\omega}^{ML}-\frac{1}{3}\Lambda\underline{e}^K\wedge\underline{e}^L\right)\\
&-\epsilon_{IJKL}\mathrm{d}\tau\wedge \underline{e}^J\wedge\left(-\mathrm{d}\omega^{KL}_{\mathrm{t}}-\underline{\omega}^K_{\phantom{K}M}\omega^{ML}_{\mathrm{t}}
+\underline{\omega}^{LM}\omega_{\mathrm{t}M}^{\phantom{M}K}-\frac{2}{3}\Lambda e^K_{\mathrm{t}}\underline{e}^L\right)=0\,.
\end{align*}
Pulling this back to $\Sigma_\tau$ we get
\begin{equation}\label{lig2}
\epsilon_{IJKL}\jmath_\tau^*\underline{e}^J\wedge\left(\mathrm{d}\jmath_\tau^*\underline{\omega}^{KL}+\jmath_\tau^*\underline{\omega}^K_{\phantom{K}M}\wedge \jmath_\tau^*\underline{\omega}^{ML}-\frac{1}{3}\jmath_\tau^*\Lambda\jmath_\tau^*\underline{e}^K\wedge\jmath_\tau^*\underline{e}^L\right)=0\,,
\end{equation}
and pulling back to to $\Sigma_\tau$ after taking the interior product with $\imath_{\partial_\tau}$ we find
\begin{align}
&-\epsilon_{IJKL}\jmath_\tau^*\underline{e}^J\wedge\frac{\mathrm{d}}{\mathrm{d}\tau}(\jmath_\tau^*\underline{\omega}^{KL})+\epsilon_{IJKL}\jmath_\tau^*e^J_{\mathrm{t}}\left(
\mathrm{d}\jmath_\tau^*\underline{\omega}^{KL}+\jmath_\tau^*\underline{\omega}^K_{\phantom{K}M}\wedge \jmath_\tau^*\underline{\omega}^{ML}-\frac{1}{3}\jmath_\tau^*\Lambda\jmath_\tau^*\underline{e}^K\wedge\jmath_\tau^*\underline{e}^L\right)\nonumber\\
&+\epsilon_{IJKL}\jmath_\tau^*\underline{e}^J\wedge\left(\mathrm{d}\jmath_\tau^*\omega^{KL}_{\mathrm{t}}+\jmath_\tau^*\underline{\omega}^K_{\phantom{K}M}\jmath_\tau^*\omega^{ML}_{\mathrm{t}}
-\jmath_\tau^*\underline{\omega}^{LM}\jmath_\tau^*\omega_{\mathrm{t}M}^{\,\phantom{M}K}+\frac{2}{3}\jmath_\tau^*\Lambda\jmath_\tau^*e^K_{\mathrm{t}}\jmath_\tau^*\underline{e}^L\right)=0\,.
\end{align}
After some straightforward manipulations and using \eqref{sust1} together with
\begin{equation}
\jmath_\tau^*\pounds_{\partial_\tau}\underline{\omega}^I_{\phantom{I}J}=\frac{\mathrm{d}}{\mathrm{d}\tau}(\jmath_\tau^*\underline{\omega}^I_{\phantom{I}J})\rightarrow Z^I_{\omega J}\,,
\end{equation}
these expressions translate into
\begin{align*}
&\epsilon_{IJKL}e^J\wedge\left(F^{KL}-\frac{1}{3}\Lambda e^K\wedge e^L\right)=0\,,\\
&\epsilon_{IJKL}\Big(e^J\wedge\left(Z_\omega^{KL}-D\omega^{KL}_{\mathrm{t}}-\Lambda e^K_{\mathrm{t}}e^L\right)-e^J_{\mathrm{t}}F^{KL}\Big)=0\,.
\end{align*}

Considering now \eqref{fieldeq_unimodularity} and proceeding as in the previous cases, we get
\[
\jmath^*_\tau\imath_{\partial_\tau}(\Phi^*\mathsf{vol})-\frac{1}{3}\jmath^*_\tau (\epsilon_{IJKL} e_{\mathrm{t}}^I e^J\wedge e^K\wedge e^L)=0\,,
\]
which, taking into account \eqref{Zvol}, the previous substitutions, and
\begin{equation}\label{sust3}
\Phi_\tau\rightarrow X\,,\qquad \dot{\Phi}_\tau\rightarrow Z_X\,,
\end{equation}
leads to \eqref{condition1}.

Finally, the same procedure applied to \eqref{fieldeq_Lambda} gives
\begin{align*}
&\mathrm{d}\jmath_\tau^*\Lambda=0\,,\\
&\frac{\mathrm{d}}{\mathrm{d}\tau}\jmath_\tau^*\Lambda=0\,,
\end{align*}
which immediately translate --taking into account that we have $Z_X^\perp\neq0$-- into $\mathrm{d}\Lambda=0$ and $Z_\Lambda=0$ by using
\begin{equation}
\jmath_\tau^*\Lambda\rightarrow\Lambda\,,\qquad \frac{\mathrm{d}}{\mathrm{d}\tau}(\jmath_\tau^*\Lambda)\rightarrow Z_\Lambda\,.
\end{equation}

As we can see, we have been able to obtain \emph{the full set} of constraints and equations for the Hamiltonian vector fields found by using the GNH method. Although we do not expect this to happen always (i.e. not all the constraints may be obtained by pulling back the field equations), by enforcing appropriate tangency requirements it should be possible to arrive at the same final description given by the GNH procedure, once the suitable presymplectic form (obtained by pulling back the canonical symplectic structure to the primary constraint submanifold) is included. This could be a convenient, alternative method to find the Hamiltonian description for singular field theories linear in time derivatives.

%
%
%

\section{Some reflections on the existing literature}\label{sec_appraisal}

The purpose of this section is to explain how our paper fits in the extensive literature on the Hamiltonian formulation for general relativity as derived from the Holst action \cite{Holst}. In the more than twenty five years since the publication of Holst's paper, a number of authors have looked at this question from different perspectives. Although we do not intend to be exhaustive, we will try to mention the most representative papers and compare their results with ours when relevant.

As a general comment, the main difference between the approach that we have followed here and the vast majority of the works on this subject stems from our use of the GNH method instead of Dirac's. As a consequence, our point of view is ``much more geometric''. Although geometry plays a role also in Dirac's approach --for instance in the classification of constraints as first or second class--, and satisfactory geometrizations of Dirac procedure already exist \cite{Diracnos}, the essence of the procedure relies heavily on the interpretation of the Poisson brackets as generators of time evolution. In contrast with this, the central consistency requirement in the GNH method is the \emph{tangency} of the Hamiltonian vector fields to the constraint submanifold.  If boundaries are present this criterion is especially appropriate. At variance with Dirac's approach, Hamiltonian vector fields play a central role because the main goal of the Dirac method is, usually, to find the constraints and study their Poisson brackets.

Works can be classified according to several features: their use or not of the time gauge, the implementation of the full $SO(1,3)$ symmetry or only the $SU(2)$ one, the treatment of second class constrains (solving them or not) and the more or less strict adherence to the Dirac algorithm as originally formulated. In many cases, the discussion of the role of the Immirzi parameter also plays a central role. A rough classification of papers according to these criteria is the following:

\begin{itemize}
\item Papers where the Hamiltonian analysis relies heavily on the use of the time gauge: \cite{Holst,Noui,Perez}.
\item Papers where no time gauge is necessary for the Hamiltonian analysis: \cite{Barros, Alexandrov1,Alexandrov2, Montesinos3}.
\item Papers where the $SO(1,3)$ invariance is explicit: \cite{Alexandrov1,Alexandrov2,Montesinos2, Montesinos3}.
\item Papers where the constraint $De^I=0$ is identified and plays a central role \cite{Noui,Cattaneo1,Cattaneo2}.
\item Papers that address the treatment of second class constraints, by avoiding its introduction, solving them or introducing Dirac brackets: \cite{Barros,Alexandrov1,Cianfrani,Montesinos1,Montesinos2,Montesinos3,Perez}.
\item The paper that, in our opinion, adheres to the letter of Dirac's algorithm in a more clear way is \cite{Perez}.
\item Works introducing constraints quadratic in momenta: \cite{Ahstekarbook, Peldan, Barros}.
\item Some incomplete analyses of the Hamiltonian formulation for unimodular gravity in terms of tetrads can be found in \cite{Smolin,Yamashita}.
\end{itemize}

The actual implementation of Dirac's algorithm is often subtle and it is important to follow it to the letter to avoid conceptual mistakes. A reason for this is the fact that secondary constraints may appear as conditions for the stability of the primary constraints or as conditions for the stability of other secondary constraints (that can show up as consistency conditions for the solvability of the equations for the Lagrange multipliers introduced in the definition of the total Hamiltonian). The analysis of the latter, in particular, is often unpleasant as they tend to be quite complicated. This is probably the reason why the simple form of the constraints that we give here has eluded most Hamiltonian analysis of the Holst action (in particular the fine one appearing in \cite{Perez}, which follows Dirac's method to the letter). Very often, these difficulties are alleviated by introducing the time gauge. This is a standard way to arrive at the Ashtekar formulation, that can be ultimately justified by invoking the Lorentz invariance of the action and the possibility of adapting the tetrads to any given spacetime foliation. Notice, however, that from a general perspective an actual gauge fixing can only be (safely) performed \emph{after the whole set of constraints and the final form of the Hamiltonian vector fields has been determined}. This is probably one of the reasons why a number of authors have discussed the possibility of dispensing with the time gauge.

The constraints $De^I=0$ have appeared in the literature, most notably, in \cite{Noui,Cattaneo1}. In \cite{Noui}, the authors check that this condition is compatible with the dynamics defined by the Holst action by looking at the Hamiltonian dynamics in the time gauge. The approach followed in that paper relies on the geometric interpretation of $De^I=0$ as a vanishing torsion condition whose compatibility with the dynamics is verified. In our approach, that condition is obtained by using the GNH method to obtain the Hamiltonian dynamics without using the time gauge.

The paper \cite{Cattaneo1} merits special attention. There, the authors discuss the Hamiltonian description of general relativity from the Holst action by relying on a geometric method introduced by Kijowski and Tulczyjew \cite{KT}. The frame fields are subject to non-degeneracy conditions equivalent to the ones that we have been naturally compelled to introduce here. Despite the apparent differences in approach, there seems to be a clear correspondence between the results of \cite{Cattaneo1} and ours (if we leave aside the part of our analysis involving parametrization and unimodularity). We list some of them here:

\begin{itemize}
\item The splitting of the conditions contained in $De^I=0$ as \emph{structural} and \emph{residual constraints} is similar to our decomposition as $e^{[I}\wedge De^{J]}=0$ and $\D^{(ij)}=0$. Notice in particular that in order to write \eqref{secondary_4}, it is necessary to choose an internal time-like vector in order to ``make the splitting into $0$ and $i$, $j$ indices''.
\item The part of the presymplectic form on the primary constraint submanifold \eqref{omega_pullback} corresponding to the triads and the spin connection is essentially equation (4.7) of \cite{Cattaneo1}.
\item The final form of the constraints.
\end{itemize}

Our results very strongly suggest that the results described here can be found by following Dirac's approach although, arguably, the necessary computations may be quite involved. It is not that the Dirac approach leads to a complicated reduced space formulation but, rather, that it is difficult to suspect that the nice formulation furnished by the constraints (\ref{constraints_simple}) can be actually found by manipulating the expressions that appear in the implementation of the Dirac algorithm. To a certain extent, this is also true within the GNH approach. Of course, the arguments relying on the derivation of the constraints from the field equations are very helpful in this regard and provide a very useful guide.

%
%
%

\section{Comments and conclusions}\label{sec_conclusions}

\noindent We have studied in detail the Hamiltonian formulation for parametrized unimodular gravity derived from a suitable modification of the Holst action on a 4-dimensional manifold diffeomorphic to $\mathbb{R}\times\Sigma$, with $\Sigma$ closed. We have relied on the GNH method which, owing to its clear geometric foundations, is superior to the more traditional Dirac approach, at least for the purposes of this paper. This is, probably, one of the reasons why we have been able to find a simple way to describe the purely gravitational sector of the theory.

The unimodularity condition has been incorporated in the action by introducing a background volume form (associated with a fixed background metric) and demanding it to be equal to the volume form defined by the tetrads. The presence of this geometric background structure has allowed us to parametrize the model in a non-trivial way by introducing dynamical diffeomorphisms. As we have shown the resulting theory reproduces the behavior of metric parametrized unimodular gravity.

It is somehow surprising that arriving at the concise formulation that we have discussed here requires a non-negligible effort, even when the GNH approach is used. However, in hindsight, it is obvious that such a formulation must exist, as shown by the argument presented in section \ref{sec_streamlined}. As we have shown, the secondary constraints \eqref{secondary_1} can be simplified to the form \eqref{secondary_4}, precisely when the non-degeneracy of the tetrads holds. This is a very neat and sensible result which is probably harder to arrive at by using Dirac's approach.

As described in detail in \cite{nos}, the real Ashtekar formulation for general relativity can be readily derived from the results presented here by using the time gauge fixing. This is, of course, to be expected as the Holst action is known to lead to the Ashtekar formulation. In a sense, the formulation described here can be thought of as the Lorentz invariant precursor of the real Ashtekar formulation. By itself, it has some interesting features:
\begin{itemize}
\item The constraints have a very simple form in the Lorentzian case and are independent of the Immirzi parameter $\gamma$.
\item Only the (pre)symplectic form depends on $\gamma$. Although the presymplectic form is arguably more complicated than the usual one (in fact, it is not written in canonical form), it is not inconceivable that it can be used for quantization. This is an interesting problem that should be looked at.
\item The internal symmetry group in this case is the full Lorentz group. The fact that it is not compact may present technical difficulties for quantization.
\item The Hamiltonian formulation for the Palatini action can be easily found from this one in the $\gamma\rightarrow\infty$ limit.
\item There is no need to introduce constraints quadratic in momenta. Actually, only the primary constraints that appear when the fiber derivative is computed involve the momenta and they do that in a very special way. First, they are linear in momenta and, second, they do not involve the velocities (this is a consequence of the fact that the Lagrangian is linear in the velocities because the action depends linearly on derivatives).
\end{itemize}

An interesting side-product of the present work is the idea of relying on the field equations to arrive at the Hamiltonian formulation, at least for first order theories and theories with actions linear in velocities such as the one discussed here. We would like to add that this is hardly a new idea. Similar ideas can be found in the literature; for instance, it is well known that the equations of motion provide \emph{Lagrangian constraints} for singular Lagrangian systems \cite{GN, DHM} and, among them, the so called \emph{projectable constraints} can be taken to the cotangent bundle of the configuration space as the usual constraints in the Hamiltonian framework (see the preceding papers and also \cite{Batlle, DM}). This method may prove to be specially fruitful when dealing with boundaries in generalizations of gravitational actions written in terms of tetrads.

%
%
\section*{Acknowledgments}
The authors wish to thank Marc Basquens for useful comments. This work has been supported by the Spanish Ministerio de Ciencia Innovaci\'on y Universidades-Agencia Estatal de Investigaci\'on/FIS2017-84440-C2-2-P grant. Bogar Díaz was partially supported by a DGAPA-UNAM postdoctoral fellowship and acknowledges support from the CONEX-Plus program funded by Universidad Carlos III de Madrid and the European Union's Horizon 2020 program under the Marie Sklodowska-Curie grant agreement No. 81538. Juan Margalef-Bentabol is supported by the Eberly Research Funds of Penn State, by the NSF grant PHY-1806356 and by the Urania
Stott fund of Pittsburgh foundation UN2017-92945.

%
%
\appendix

\section{Glossary and Notation}\label{appendix_notation}

Although we have introduced some of the notation used in the paper whenever it was relevant, in order to facilitate its reading we have left a number of definitions for this appendix.

\medskip

\noindent $\bullet$ First, given a volume form $\mathrm{vol}$ in a differentiable manifold $\mathcal{M}$ and a top-form $\alpha$, it is always possible to find a smooth function $f\in C^\infty(\mathcal{M})$ such that $\alpha=f\cdot \mathrm{vol}$. We will often denote such function as
\[
\left(\frac{\alpha}{\mathrm{vol}}\right)\,.
\]
Although this seems cumbersome at first, the unwieldy parentheses are a good reminder of the fact that we are dealing with a scalar density.

\medskip

\noindent $\bullet$ Second, we introduce now the simplified notation used to write the solutions to the equations for $Z_e^I$ and $Z_e^0$ and in the tangency analysis. Let $w$ be the volume form on $\Sigma$
\[
w=\frac{1}{3!}\epsilon_{ijk}e^i\wedge e^j\wedge e^k\,,
\]
(remember that we are working with non-degenerate frames in $\Sigma$). We define
\begin{align*}
&\E^{ij}:=\left(\frac{e^i\wedge e^j\wedge e^0}{w}\right) \,, && \\
&\D^{\circ i}:=\left(\frac{(De^0)\wedge e^i}{w}\right)\,, &&\D^{i\circ}:=\left(\frac{(De^i)\wedge e^0}{w}\right)\,, &&\D^{ij}:=\left(\frac{e^i\wedge De^j}{w}\right)\,,\\
&\F^{i\circ j}:=\left(\frac{e^i\wedge F^{0j}}{w}\right)\,,&& \F^{\circ ij}:=\left(\frac{e^0\wedge F^{ij}}{w}\right)\,,&&\\
&\F^\circ_{\,\,\,\circ i}:=\left(\frac{e^0\wedge F_{0i}}{w}\right)\,,&& \F^{ijk}:=\left(\frac{e^i\wedge F^{jk}}{w}\right)\,.&&
\end{align*}
Notice that $\E^{ij}$, $\F^{\circ ij}$, $\F^{ijk}$, $\F^{i\circ j}$ and $\F^\circ_{\,\,\,\circ i}$ are antisymmetric in the last pair of indices.

%
%
\section{Diffeomorphisms and embeddings as dynamical variables}\label{appendix_diffeos}

We give here a few details about the use of diffeomorphisms and embeddings as dynamical variables. Interested readers are referred to \cite{margalef2018thesis,parame_scalar,EM} for details. Given the 4-dimensional manifolds $\mathbb{R}\times\Sigma$ and $\mathcal{M}$, we use as dynamical variables diffeomeorphisms $\Phi:\mathbb{R}\times\Sigma\rightarrow \mathcal{M}$ such that for every $\tau\in\mathbb{R}$ the embeddings $\Phi_\tau:=\Phi\circ\jmath_\tau:\Sigma\hookrightarrow\mathcal{M}$ have spacelike images (i.e. $\Phi_\tau(\Sigma)\subset\mathcal{M}$ is $g$-spacelike). Here $\jmath_\tau:\Sigma\rightarrow\mathbb{R}\times\Sigma:p\mapsto (\tau,p)$. We will denote the space of such embeddings as $\mathrm{Emb}_s(\Sigma,\mathcal{M})$. The diffeomorphisms that we consider in the paper can be loosely interpreted as curves of embeddings of this type.

We denote the tangent map associated with $\Phi$ as $T\Phi$. The diffeomorphisms that we use are such that the field $\dot{\Phi}:=T\Phi.\partial_\tau$ is transverse to $\Phi_\tau(\Sigma)$ for all $\tau\in\mathbb{R}$. When restricted to a particular embedding $\Phi_\tau$, $\dot{\Phi}$ defines its instantaneous velocity $\dot{\Phi}_\tau$.

Given $X\in\mathrm{Emb}_s(\Sigma,\mathcal{M})$ we can build a vector field over $X$ consisting of future directed, unit normals that we denote as $\mathbf{n}_X$. Notice that $\mathbf{n}_X(p)\in T_{X(p)}\mathcal{M}$ for each $p\in\Sigma$. Now, if we have $Y_X\in\Gamma(X^*T\mathcal{M})$ we can expand it as
\[
Y_X=Y_X^\perp \mathbf{n}_X+TX.Y_X^\top\,,
\]
where $Y_X^\perp=g(\mathbf{n}_X,Y_X)=:n_X(Y_X)$ is a smooth real function on $\Sigma$, $Y^\top_X\in\mathfrak{X}(\Sigma)$ and $TX$ is the tangent map of $X$.

A useful result --that we give without proof-- is
\begin{equation}\label{Zvol}
\jmath_\tau^*\imath_{\partial_\tau}(\Phi^*\mathsf{vol})=\dot{\Phi}^\perp_\tau\mathsf{vol}_{\gamma_{\Phi_\tau}}\,,
\end{equation}
where $\dot{\Phi}^\perp_\tau=\varepsilon n_{\Phi_\tau}(\dot{\Phi}_\tau)$ and $\mathsf{vol}_{\gamma_{\Phi_\tau}}\!\in\Omega^3(\Sigma)$ is the volume form associated with the metric $\gamma_{\Phi_\tau}:=\Phi^*_\tau g$ on $\Sigma$:
\[
\mathsf{vol}_{\gamma_{\Phi_\tau}}(v_1,v_2,v_3):=\mathsf{vol}(\mathbf{n}_{\Phi_\tau},T\Phi.v_1,T\Phi.v_2,T\Phi.v_3)\,,
\]
for $v_1\,,v_2\,,v_3\in\mathfrak{X}(\Sigma)$. Equation \eqref{Zvol} is used in sections \ref{subsec_lagrangian} and \ref{sec_streamlined}.

%
%
\section{Useful mathematical results}\label{appendix_useful_results}

\noindent $\bullet$ The invariant $SO(1,3)$ tensor $P_{IJKL}$ can be inverted whenever $\gamma^2\neq \varepsilon$ in the sense that $P_{IJKL}[P^{-1}]^{KLMN}=\delta_I^{\,\,[M}\delta_J^{\,\,N]}$. The inverse is
\begin{equation}\label{Pinverse}
[P^{-1}]^{IJKL}:=\frac{\gamma}{2(\varepsilon-\gamma^2)}\left(-\varepsilon\gamma\cdot\epsilon^{IJKL}+\eta^{IK}\eta^{JL}-\eta^{JK}\eta^{IL}\right)\,,
\end{equation}
as can be checked by a direct computation.

\medskip

\noindent $\bullet$ The tensor $P_{IJKL}$ satisfies $DP_{IJKL}=0$.

\medskip

\noindent $\bullet$ For any $H^{IJ}\in C^\infty(\Sigma)$ antisymmetric in $I$ and $J$ the following identity holds
\begin{equation}\label{identity1}
[P^{-1}]^{IJ}_{\,\,\,\,\,\,NQ}P^{NK}_{\,\,\,\,\,\,\,\,\,\,\,\,LM}H^Q_{\,\,\,\,\,K}=H^{[I}_{\,\,\,\,\,[L}\delta^{J]}_{\,\,\,\,\,M]} \,.
\end{equation}

\medskip

\noindent \textbf{Proof.}
\begin{equation*}
[P^{-1}]^{IJ}_{\,\,\,\,\,\,\,NQ}P^{NK}_{\,\,\,\,\,\,\,\,\,\,LM}H^Q_{\,\,\,\,K}=H^{[I}_{\,\,\,\,\,[L}\delta^{J]}_{\,\,\,\,\,M]}+\frac{\gamma}{2(\varepsilon-\gamma^2)}\big( \epsilon^{IJ}_{\,\,\,\,\,\,\,K[L}H^K_{\,\,\,\,M]}-\epsilon^{K[I}_{\,\,\,\,\,\,\,\,\,\,\,LM}H^{J]}_{\,\,\,\,\,\,\,K}\big)\,.
\end{equation*}
The last term of this expression can be shown to vanish by computing its dual
\begin{align*}
&\epsilon^{NQ}_{\quad\,\,IJ}\big(\epsilon^{IJ}_{\,\,\,\,\,\,\,K[L}H^{K}_{\,\,\,M]}-\epsilon^{KI}_{\,\,\,\,\,\,\,\,\,LM}H^{J}_{\,\,\,K}\big)\\
&=2\varepsilon\Big(\delta^N_K\delta^Q_{[L}H^{K}_{\,\,\,M]}-\delta^Q_K\delta^N_{[L}H^{K}_{\,\,\,M]}\big)\\
&+\varepsilon\big(\delta^J_K\delta^N_L\delta^Q_M+\delta^N_K\delta^Q_L\delta^J_M+\delta^Q_K\delta^J_L\delta^N_M
-\delta^N_K\delta^J_L\delta^Q_M-\delta^J_K\delta^Q_L\delta^N_M-\delta^Q_K\delta^N_L\delta^J_M\big)H_J^{\,\,\,K}=0\,.\quad\qquad\blacksquare
\end{align*}

\medskip

\noindent$\bullet$ In the following we give the solutions to several types of inhomogeneous linear equations involving differential forms in a 3-dimensional manifold $\Sigma$. We use the notation explained in appendix \ref{appendix_notation}. The proofs are quite direct so they are left to the reader.

\medskip

\noindent i) Let us consider the system of equations
\begin{equation}\label{equation-1}
\epsilon_{ijk}e^j\wedge z^k=u_i \,,
\end{equation}
where the unknowns are $z^k\in\Omega^1(\Sigma)$, with $\Sigma$ a three-dimensional manifold, the triads $e^i\in\Omega^1(\Sigma)$ are such that $w:=(\epsilon_{ijk}e^i\wedge e^j\wedge e^k)/3!$ is a volume form in $\Sigma$ and the $u^i\in\Omega^2(\Sigma)$ are given 2-forms. Then the solutions are 
\begin{equation}\label{sols1}
z^k=\frac{1}{2}\left(\frac{e^i\wedge u_i}{w}\right)e^k-\left(\frac{e^k\wedge u_i}{w}\right)e^i\,.
\end{equation}
Even though \eqref{equation-1} is inhomogeneous it can be solved for any given $u_i$.
\medskip

\noindent ii) Let us consider the system of equations
\begin{equation}\label{equation-2}
\alpha\wedge e^i=\beta^i \,,
\end{equation}
where $\alpha\in\Omega^1(\Sigma)$ is the unknown, $e^i\in\Omega^1(\Sigma)$ are such that $w:=(\epsilon_{ijk}e^i\wedge e^j\wedge e^k)/3!$ is a volume form in $\Sigma$, and the $\beta^i\in\Omega^2(\Sigma)$ are given 2-forms. Equation \eqref{equation-2} can be solved if and only if the inhomogeneous term $\beta^i$ satisfies the condition
\begin{equation}\label{condition_alpha}
\beta^i\wedge e^j+\beta^j\wedge e^i=0\,,
\end{equation}
in which case the solution is
\begin{equation}\label{sols2}
\alpha=\frac{1}{2}\epsilon_{ijk}\left(\frac{\beta^i\wedge e^j}{w}\right)e^k\,.
\end{equation}

\medskip

\noindent iii) Let us consider the equation
\begin{equation}\label{equation-3}
e^i\wedge z_i=u \,,
\end{equation}
where the unknowns are $z_k\in\Omega^1(\Sigma)$, the triads $e^i\in\Omega^1(\Sigma)$ are such that $w:=(\epsilon_{ijk}e^i\wedge e^j\wedge e^k)/3!$ is a volume form in $\Sigma$ and $u\in\Omega^2(\Sigma)$ is a given 2-form. Then the solutions are
\begin{equation}\label{sols3}
z_i=-\frac{1}{2}\epsilon_{ijk}\left(\frac{e^j\wedge u}{w}\right)e^k+\zeta_{ij}e^j\,,
\end{equation}
with $\zeta_{ij}\in C^\infty(\Sigma)$ satisfying $\zeta_{ij}=\zeta_{ji}$ but, otherwise, arbitrary.

\medskip

\noindent$\bullet$ The following identity is useful 
\begin{equation}\label{Id_E}
e^0=\frac{1}{2}\epsilon_{ijk}\E^{ij}e^k \,.
\end{equation}
To prove it we expand $e^0=\lambda_ie^i$ and plug it into
\[
\epsilon_{jkl}\E^{jk}e^l=\epsilon_{jkl}\left(\frac{e^j\wedge e^k\wedge e^0}{w}\right)e^l \,,
\]
to get
\[
\epsilon_{jkl}\E^{jk}e^l=2\lambda_i e^i=2e^0\,.
\]
With the help of \eqref{Id_E}, it is straightforward to prove the following relations
\begin{subequations}
\begin{align}
&\D ^{i\circ}=\frac{1}{2}\epsilon_{klm}\E^{kl}\D^{mi}\,,\label{rel_2}\\
&\F^{\circ ij}=\frac{1}{2}\epsilon_{klm}\E^{kl}\F^{mij}\,,\label{rel_3}\\
&\F^\circ_{\,\,\,\circ i}=\frac{1}{2}\epsilon_{jkl}\E^{jk}\F^l_{\,\,\,\circ i}\,.\label{rel_4}
\end{align}
\end{subequations}
Other useful identities are
\begin{subequations}
\begin{align}
& F^{ij}=\frac{1}{2}\epsilon_{klm}\F^{kij}e^l\wedge e^m\,,\label{rel_5}\\
& F^{0i}=\frac{1}{2}\epsilon_{jkl}\F^{j\circ i}e^k\wedge e^l\,,\label{rel_6}\\
& De^i=\frac{1}{2}\epsilon_{jkl}\D^{ji}e^k\wedge e^l\,,\label{rel_7}\\
& De^0=\frac{1}{2}\epsilon_{ijk}\D^{\circ i}e^j\wedge e^k\,.\label{rel_8}
\end{align}
\end{subequations}

\noindent These identities are all proven in the same way, so we will just show that \eqref{rel_7} holds.

\begin{prop}
$De^i=\frac{1}{2}\epsilon_{jkl}\D^{ji}e^k\wedge e^l$ is equivalent to $\D^{ij}=\left(\frac{e^i\wedge De^j}{w}\right)$.
\end{prop}

\begin{proof}

\medskip

\noindent \boxed{\Rightarrow} $De^i=\frac{1}{2}\epsilon_{jkl}\D^{ji}e^k\wedge e^l$ implies $e^m\wedge De^i=\frac{1}{2}\epsilon_{jkl}\D^{ji}e^m\wedge e^k\wedge e^l=\frac{1}{2}\epsilon_{jkl}\D^{ji}\epsilon^{mkl}w=\D^{mi}w$, whence, $\D^{ij}=\left(\frac{e^i\wedge De^j}{w}\right)$.

\medskip

\noindent \boxed{\Leftarrow} Let us write $De^i=\frac{1}{2}\epsilon_{jkl}H^{ji}e^k\wedge e^l$, then $e^m\wedge De^i=\frac{1}{2}\epsilon_{jkl}H^{ji}e^m\wedge e^k\wedge e^l=\frac{1}{2}\epsilon_{jkl}H^{ji}\epsilon^{mkl}w=H^{mi}w$ which implies $H^{ji}=\D^{ji}$, so that $De^i=\frac{1}{2}\epsilon_{jkl}\D^{ji}e^k\wedge e^l$.\end{proof}

\noindent $\bullet$ $F^I_{\phantom{I}J}\wedge e^J=0$ is equivalent to
\begin{subequations}
\begin{align}
&\F^j_{\phantom{j}\circ j}=0\,,\label{FF1}\\
&\F^{\circ i}_{\phantom{\circ i}\circ}+\F^{ji}_{\phantom{ji}j}=0\,.\label{FF2}
\end{align}
\end{subequations}

\subsection{Rewriting the non-degeneracy condition for the tetrads}\label{Ap_non_deg}

It is interesting to take a close look at the condition \eqref{condition_ev}. To begin with, it is important to notice that, in terms of the objects defined in the preceding section, it can be written as
\begin{equation}\label{condition_ev_new}
e_{\mathrm{t}}^0-\frac{1}{2}\epsilon_{klm}e_{\mathrm{t}}^k\E^{lm}\neq0\,,
\end{equation}
because
\[
\varepsilon_{IJKL}e_{\mathrm{t}}^I e^J\wedge e^K\wedge e^L=6\left(e_{\mathrm{t}}^0-\frac{1}{2}\epsilon_{klm}e_{\mathrm{t}}^k\E^{lm}\right)w\,.
\]
The meaning of \eqref{condition_ev_new} as a non-degeneracy condition for the tetrads is obvious if we write them in matrix form as
\begin{equation}\label{tetrad_non_deg}
e^I=\left[\begin{array}{r|r}
e_{\mathrm{t}}^0&e_{\mathrm{t}}^i\\\hline
\phantom{\big|\!\!} e^0&e^i
\end{array}\right] \,,
\end{equation}
and remember that 
\begin{equation}
e^0=\frac{1}{2}\epsilon_{ijk}e^i\E^{jk}=:\lambda_i e^i\,.\label{lambda}
\end{equation}
From \eqref{tetrad_non_deg}, we can write the 4-metric as
\begin{equation}\label{tetrad_non_deg2}
g=\left[\begin{array}{c|c}
\varepsilon e_{\mathrm{t}}^0e_{\mathrm{t}}^0+e_{\mathrm{t}i}e_{\mathrm{t}}^i&\varepsilon e_{\mathrm{t}}^0 e^{0\top}+e_{\mathrm{t}i} e^{i\top}\\\hline
\phantom{\big|\!\!} \varepsilon e_{\mathrm{t}}^0 e^0+e_{\mathrm{t}i} e^i&\varepsilon  e^0e^{0\top}+e_ie^{i\top}
\end{array}\right]\,,
\end{equation}
and conclude that a necessary and sufficient condition for the 3-metric $q:=\varepsilon e^0e^{0\top}+e_ie^{i\top}$ to be definite positive is $1+\varepsilon \lambda_i\lambda^i>0$. This can be seen by writing $q_{ab}=(\delta_{ij}+\varepsilon \lambda_i\lambda_j){e_a}^{\!i}{e_b}^{\!j}$ and noting that the quadratic form $\delta_{ij}+\varepsilon \lambda_i\lambda_j$ is positive definite if and only if $1+\varepsilon \lambda_i\lambda^i>0$.

\subsection{A determinant computation}\label{Ap_det}

We compute here the determinant of the $6\times6$ matrix obtained by grouping the indices of
\[
M^{(ij)}_{\phantom{(ij)}(kl)}:=\delta^i_k\delta^j_l-\delta^{ij}\delta_{kl}+\varepsilon \E_k^{\phantom{k}(i}\E^{j)}_{\phantom{j)}l} \,,
\]
as the symmetrized pairs $(ij)$ and $(kl)$.  To simplify the computations and interpret the result, it is useful to write $M^{(ij)}_{\phantom{(ij)}(kl)}$ in terms of the $\lambda_i$ defined in \eqref{lambda} as
\[
\delta^{(i}_k\delta^{j)}_l(1+\varepsilon \lambda_m\lambda^m)-\delta^{ij}\delta_{kl}(1+\varepsilon \lambda_m\lambda^m)+\varepsilon \delta^{ij}\lambda_k\lambda_l+\varepsilon \lambda^i\lambda^j\delta_{kl}-\varepsilon\lambda^{(i}\delta^{j)}_k\lambda_l-\varepsilon\lambda^{(i}\delta^{j)}_l\lambda_k\,.
\]
The determinant of $M^{(ij)}_{\phantom{(ij)}(kl)}$ can be obtained with the help of any computer algebra package. It has the simple expression
\[
-\frac{1}{4}\left(1+\varepsilon\lambda_m\lambda^m\right)^2\,,
\]
which is equivalent to
\[
-\frac{1}{4}\left(1+\frac{\varepsilon}{2}\E_{ij}\E^{ij}\right)^2\,.
\]
As can be seen from the discussion in section \ref{Ap_non_deg}, the condition that the 3-metric $\varepsilon e^0e^{0\top}+e_ie^{i\top}$ be positive definite implies that the determinant of $M^{(ij)}_{\phantom{(ij)}(kl)}$ is different from zero.

%
%
\section{Some details about the GNH procedure}\label{appendix_GNH}

\subsection{Computing the pullback of $\Omega$ to $\mathfrak{M}_0$}\label{Ap_GNH_pullback}

In order to do this, we have to compute \eqref{symplectic_form} for vector fields $\mathbb{Y}\,,\mathbb{Z}$ which are tangent to the primary constraint submanifold $\mathfrak{M}_0$. Some of the components of these fields on $\mathfrak{M}_0$ can be obtained from the primary constraints by computing their directional derivatives and requiring them to be zero, so by plugging these particular components into the canonical symplectic form we get the sought for pull-back. The only unfamiliar directional derivatives are those involving embedding-dependent objects (which only show up in the definition of ${\bm{\mathrm{p}}}_{\mathsf{X}}$). In order to compute them, one has to use variations as in \cite{parame_scalar}. In the present case, from the definition of the momenta \eqref{momenta} we get
\begin{subequations}\label{ZcomponentsM0}
\begin{align}
&{\bm{\mathrm{p}}}_{e_{\mathrm{t}}}&&\hspace*{-3mm}\longrightarrow \hspace*{-2mm}&&{\bm{Z}}_{e_{\mathrm{t}}}(\cdot)=0 \,,\\
&{\bm{\mathrm{p}}}_e&&\hspace*{-3mm}\longrightarrow \hspace*{-2mm}&&{\bm{Z}}_e(\cdot)=0 \,,\\
&{\bm{\mathrm{p}}}_{\omega_{\mathrm{t}}}&&\hspace*{-3mm}\longrightarrow \hspace*{-2mm}&&{\bm{Z}}_{\omega_{\mathrm{t}}}(\cdot)=0 \,,\\
&{\bm{\mathrm{p}}}_\omega&&\hspace*{-3mm}\longrightarrow \hspace*{-2mm}&&{\bm{Z}}_\omega(\cdot)=\!\!\int_\Sigma 2P_{IJKL}(\cdot)\wedge Z_e^K\wedge e^L \,,\\
&{\bm{\mathrm{p}}}_\Lambda&&\hspace*{-3mm}\longrightarrow \hspace*{-2mm}&&{\bm{Z}}_\Lambda(\cdot)=0\,,\\
&{\bm{\mathrm{p}}}_{\mathsf{X}}&&\hspace*{-3mm}\longrightarrow \hspace*{-2mm}&&{\bm{Z}}_{\mathsf{X}}(\cdot)=\!\!\int_\Sigma \!\Big(\!\!-\Lambda\pounds_{e_X(\cdot)}Z_{\mathsf{X}}^\perp+\varepsilon Z_\Lambda n_X(\cdot)+\varepsilon n_X(\cdot)\Lambda\mathrm{div}_{\gamma_X}Z_{\mathsf{X}^\top}\!\Big)\mathsf{vol}_{\gamma_X} \,,
\end{align}
\end{subequations}
where $n_X(\cdot)$ and $e_X(\cdot)$ are embedding dependent objects (which are carefully discussed in \cite{parame_scalar}). 

If we demand that the components of $\mathbb{Z}\in \mathfrak{X}(T^*Q)$ take the values given by \eqref{ZcomponentsM0}, the resulting vector field will be tangent to $\mathfrak{M}_0$ so its restriction to $\mathfrak{M}_0$, that will be denoted as $\mathbb{Z}_0$, will be a vector field on $\mathfrak{M}_0$.

Taking into account that $Y^\alpha_{\mathsf{X}}=Y_{\mathsf{X}}^\perp n_X^\alpha+ (TX)_a^\alpha Y_{\mathsf{X}}^{\top a}$ (here $(TX)_a^\alpha$ denotes the tangent map of the embedding $X$ and the indices $a$ and $\alpha$ refer to $\Sigma$ and $\mathcal{M}$ respectively) we see that
\begin{align*}
&\varepsilon n_X(Y_{\mathsf{X}})=Y_{\mathsf{X}}^\perp\,,\\
&e_X(Y_{\mathsf{X}})=Y_{\mathsf{X}}^\top\,,
\end{align*}
and, hence,
\[
\bm{\mathrm{Z}}_{\mathsf{X}}(Y_{\mathsf{X}})=\int_\Sigma\left(-\Lambda\pounds_{Y_{\mathsf{X}}^\top}Z_{\mathsf{X}}^\perp-\pounds_{Z_{\mathsf{X}}^\top}(\Lambda Y_{\mathsf{X}}^\perp)+Z_\Lambda Y_{\mathsf{X}}^\perp\right)\mathsf{vol}_{\gamma_X} \,.
\]
Plugging the previous results into \eqref{symplectic_form}, we finally get
\begin{align*}\label{symplectic_pullback}
\omega(\mathbb{Z}_0,\mathbb{Y}_0)=&\int_\Sigma \Big(2P_{IJKL}\big(Z_\omega^{IJ}\wedge Y_e^K-Y_\omega^{IJ}\wedge Z_e^K\big)\wedge e^L\\
&\hspace*{2cm}+Z_{\mathsf{X}}^\perp\big(Y_\Lambda-\imath_{Y_{\mathsf{X}}^\top}\mathrm{d}\Lambda\big)\mathsf{vol}_{\gamma_X}
-Y_{\mathsf{X}}^\perp\big(Z_\Lambda-\imath_{Z_{\mathsf{X}}^\top}\mathrm{d}\Lambda\big)\mathsf{vol}_{\gamma_X}\Big) \,,
\end{align*}
which, being closed but degenerate, is a presymplectic form on $\mathfrak{M}_0$.

\subsection{Rewriting the new secondary constraints \eqref{secondary_1}}\label{Ap_GNH_rewritingnew}

The constraints \eqref{constraint1_ij_bis} and \eqref{constraint1_i0_bis} can be used to rewrite \eqref{secondary_1} in the much simpler and suggestive form \eqref{secondary_2}. To this end, we first use \eqref{constraint1_ij_bis} and \eqref{constraint1_i0_bis} to rewrite \eqref{secondary_1} as
\[
\epsilon_{klm}e_{\mathrm{t}}^k\big(\D^{il}\E^{jm}+\D^{jl}\E^{im}\big)-2e_{\mathrm{t}}^0\D^{ij}+e_{\mathrm{t}}^i \D^{j\circ}+e_{\mathrm{t}}^j \D^{i\circ}=0\,.
\]
We then use \eqref{rel_2} to transform the last two terms of the preceding expression to get
\begin{equation}\label{secondary_3}
\epsilon_{klm}e_{\mathrm{t}}^k\D^{il}\E^{jm}+\frac{1}{2}\epsilon_{klm}e_{\mathrm{t}}^j\E^{kl}\D^{mi}+\epsilon_{klm}e_{\mathrm{t}}^k\D^{jl}\E^{im}
+\frac{1}{2}\epsilon_{klm}e_{\mathrm{t}}^i\E^{kl}\D^{mj}-2e_{\mathrm{t}}^0 \D^{ij}=0\,.
\end{equation}
The obvious identity
\[
0=\epsilon_{klm}\D^{i[m}e_{\mathrm{t}}^k \E^{jl]} \,,
\]
is equivalent to
\begin{equation}\label{identDij}
\big(\epsilon_{klm}e_{\mathrm{t}}^k\E^{lm}\big)\D^{ij}-\epsilon_{klm}e_{\mathrm{t}}^j\D^{ik}\E^{lm}+2\epsilon_{klm}e_{\mathrm{t}}^k\D^{im}\E^{jl}=0\,,
\end{equation}
which, together with \eqref{constraint1_ij_bis},  allows us to simplify \eqref{secondary_3} to the form \eqref{secondary_2}.
Notice that, with the help of \eqref{constraint1_ij_bis}, we can write \eqref{identDij} also in the form
\begin{equation}\label{identDij2}
\big(\epsilon_{klm}e_{\mathrm{t}}^k\E^{lm}\big)\D^{ij}-\epsilon_{klm}e_{\mathrm{t}}^i\D^{jk}\E^{lm}+2\epsilon_{klm}e_{\mathrm{t}}^k\D^{jm}\E^{il}=0\,.
\end{equation}

\subsection{Some results useful to obtain $Z_\omega^{IJ}$}\label{Ap_GNH_Zomega}

We show that by $D$-differentiating \eqref{condition5_bis} and using \eqref{Ze_const} and the secondary constraints, equation \eqref{tan3_bis} holds. Indeed, by  differentiating \eqref{condition5_bis} and using the constraint $De^I=0$ and the Bianchi identity $DF^{IJ}=0$ we get
\[
\epsilon_{IJKL}e^J\wedge DZ^{KL}_\omega=F^K_{\phantom{K}M}\omega^{ML}_{\mathrm{t}}-F^L_{\phantom{L}M}\omega_{\mathrm{t}}^{MK}-\epsilon_{IJKL}De_{\mathrm{t}}^J(F^{KL}-\Lambda e^K\wedge e^L)\,.
\]
Plugging this now into \eqref{tan3_bis} and using \eqref{Ze_const}, we obtain
\[
\epsilon_{IJKL}\omega_{\mathrm{t}M}^{\,\phantom{M}J}\left(e^{[L}\wedge F^{MK]}-\frac{1}{3}\Lambda e^L\wedge e^M\wedge e^K\right)=0\,,
\]
which holds as a consequence of the constraint \eqref{constraints_simple3}.

%
%
\providecommand{\href}[2]{#2}\begingroup\raggedright\endgroup

\end{document}